\documentclass[journal]{IEEEtran}

\usepackage[T1]{fontenc}
\usepackage{subcaption}
\usepackage{color}
\usepackage{varioref}
\usepackage{textcomp}
\usepackage{amsthm}
\usepackage{amsmath}
\usepackage{amssymb}
\usepackage{graphicx}
\usepackage{setspace}
\usepackage{esint}
\usepackage{algorithm}
\usepackage{algpseudocode}
\usepackage{pifont}
\usepackage{amsmath}

\usepackage{enumitem}
\usepackage{cite}
\makeatletter

\newcommand{\Rmnum}[1]{\expandafter\@slowromancap\romannumeral #1@}
\makeatother

\newtheorem{theorem}{Theorem}
\newtheorem{lemma}{Lemma}
\newtheorem{remark}{Remark}

\newtheorem{thm}{\protect\theoremname}
\newtheorem{prop}[thm]{\protect\propositionname}
\providecommand{\propositionname}{Proposition}

\ifCLASSINFOpdf

\else

\fi

\usepackage[T1]{fontenc}
\usepackage{etoolbox}

\makeatletter
\patchcmd{\maketitle}{\@fnsymbol}{\@alph}{}{}  % Footnote numbers from symbols to small letters
\makeatother

\title{Caching and Coded Delivery over Gaussian Broadcast Channels for Energy Efficiency}
\author{
  Mohammad Mohammadi Amiri,~\IEEEmembership{Student Member,~IEEE}, and~Deniz~G\"und\"uz,~\IEEEmembership{Senior Member,~IEEE}
}
\date{}
\begin{document}

\maketitle

\begin{abstract}

A cache-aided $K$-user Gaussian broadcast channel (BC) is considered. The transmitter has a library of $N$ equal-rate files, from which each user demands one. The impact of the equal-capacity receiver cache memories on the minimum required transmit power to satisfy all user demands is studied. Considering uniformly random demands across the library, both the minimum average power (averaged over all demand combinations) and the minimum peak power (minimum power required to satisfy all demand combinations) are studied. Upper bounds are presented on the minimum required average and peak transmit power as a function of the cache capacity considering both \textit{centralized} and \textit{decentralized} caching. The lower bounds on the minimum required average and peak power values are also derived assuming uncoded cache placement. The bounds for both the peak and average power values are shown to be tight in the centralized scenario through numerical simulations. The results in this paper show that proactive caching and coded delivery can provide significant energy savings in wireless networks. 
\end{abstract}


\section{Introduction}\label{Intro}

We\makeatletter{\renewcommand*{\@makefnmark}{}
\footnotetext{Manuscript received December 2, 2017; revised April 6, 2018; accepted April 18, 2018. This work received support from the European Research Council (ERC) through Starting Grant BEACON (agreement 677854).}\makeatother} study proactive content caching to user devices followed by coded delivery \cite{MaddahAliCentralized}, assuming that the delivery phase takes place over a Gaussian broadcast channel (BC) from the server to the users.\makeatletter{\renewcommand*{\@makefnmark}{}
\footnotetext{The authors are with Imperial College London, London SW7 2AZ, U.K. (e-mail: m.mohammadi-amiri15@imperial.ac.uk; d.gunduz@imperial.ac.uk).}\makeatother} Several other recent papers have studied coded delivery over noisy channels. Fading and interference channel models are considered in \cite{HuangFadingChannelcodedcaching} and \cite{NaderializadehMaddahAliInterference}, respectively. In \cite{ShirinErasureChannelJournal} and \cite{MohammadDenizPacketErasureJournal}, centralized caching is considered while the delivery phase takes place over a packet-erasure BC. The capacity-memory trade-off is investigated in this setting assuming that only the weak users have caches, which requires knowledge about the channel qualities in the delivery phase in advance. Proactive caching and coded delivery over a Gaussian BC is studied in \cite{PetrosEliaTopological} in the high power regime. In \cite{ShirinWiggerYenerCacheAssingment} cache-aided data delivery is studied over a degraded BC in a centralized setting, while \cite{LeiDelaySuperPosWirelessNet} focuses on the delay analysis in a cache-aided wireless channel with coded delivery by utilizing superposition coding. A multi-antenna server is considered in \cite{EnricoHamdiBrunoISIT17,NgoYangKobayashiScalable,PooyaCaireKhalajPhysicalLayerJournal} with different assumptions regarding the channel state information available at the transmitter.

%\makeatletter{\renewcommand*{\@makefnmark}{}
%\footnotetext{Part of this work was presented at the IEEE International Symposium on Information Theory, Aachen, Germany, June 2017 [14].}\makeatother}

While most of the previous works on coded delivery over noisy wireless broadcast channels focus on the high signal-to-noise ratio (SNR) regime; in this work, our goal is to highlight the benefits of proactive caching and coded delivery from an energy efficiency perspective. The impact of proactive caching on energy efficiency was previously studied in \cite{GungorProactive} and \cite{GregoryDtoD} in an offline scenario, assuming known demands and channel conditions. However, these works consider only proactive caching and do not benefit from coded delivery.

In \cite{NaderializadehMaddahAliInterference,ShirinErasureChannelJournal,MohammadDenizPacketErasureJournal,PetrosEliaTopological,ShirinWiggerYenerCacheAssingment}, users' channel conditions in the delivery phase are assumed to be known in advance during the placement phase. However, in practical wireless networks, it is often not possible to know the identities of users that will participate in a particular delivery phase, let alone their channel conditions. Therefore, our goal here is to study the benefits of proactive caching in reducing the transmit power, assuming that the noiseless cache placement phase is carried out without the knowledge of channel conditions during the delivery phase.

Assuming uniform popularity across the files, we first study the minimum required \textit{average power} to serve all the users, averaged over all the user demand combinations. Note that we allow the transmitter to change its power depending on the demand combination in order to minimize the average power consumption. We then consider the transmit power required to satisfy the worst-case demand combination, called the \textit{peak power}. Building upon our previous work in \cite{MohammadDenizGaussianBCISIT17}, we first provide upper bounds on the minimum average and peak power values as a function of the rate of the files in the library and the capacity of the user caches, for centralized cache placement. We then extend the proposed scheme by considering \textit{decentralized} cache placement. The proposed delivery strategy employs superposition coding and power allocation, and the achievable transmit power for any demand combination is derived thanks to the degradedness of a Gaussian BC. We further derive lower bounds on the performance assuming uncoded cache placement.

The main novelty of the proposed proactive caching and coded delivery scheme is the way the coded packets designed for each user are generated for any demand combination, particularly when a file may be requested by more than one user, and the way these coded packets are delivered over a Gaussian BC in order to minimize the transmit power. We show that the proposed achievable scheme reduces the transmit power significantly, even with the availability of only a small cache capacity at each receiver, in both the centralized and decentralized scenarios. It is also shown that the power loss between the centralized and the more practical decentralized scenario is quite small. Furthermore, numerical results show that the gaps between the peak and average transmit powers of the proposed achievable scheme for the centralized scenario and the corresponding lower bounds are negligible. In particular, we have observed numerically that the multiplicative gap between the two bounds for the centralized caching is below 2 for the examples considered. Numerical results also illustrate that adjusting the transmit power based on the demand vector can significantly reduce the average power consumption compared to the worst-case demand combination.

%The rest of the paper is organized as follows. In Section \ref{SystemModel}, system model and preliminaries are introduced. The proposed caching and coded delivery scheme is presented for the centralized and decentralized settings in Section \ref{ProposedSchemeCentralized} and Section \ref{ProposedSchemeDecentralized}, respectively. In Section \ref{ProofSecondTheorem}, the converse bound is derived. Numerical results are illustrated in Section \ref{NumericalRes}. Finally the paper is concluded in Section \ref{Conc}. The details of the proofs are provided in the Appendix.    

\section{System Model and Preliminaries}\label{SystemModel}
We study cache-aided content delivery over a $K$-user Gaussian BC. The server has a library of $N$ files, $\mathbf{W} \buildrel \Delta \over = \left( W_1, ..., W_N \right)$, each distributed uniformly over the set\footnote{For any positive integer $i$, $\left[ i \right]$ denotes the set $\{ 1, ..., i \}$.} $\left[ \left\lceil 2^{nR} \right\rceil \right]$, where $R$ is the rate of the files and $n$ denotes the blocklength, referring to $n$ uses of the BC. Each user has a cache of size $nMR$ bits. We define the \textit{normalized global cache capacity} as $t \buildrel \Delta \over = MK/N$.  

Data delivery from the server to the users takes place in two phases. Caches of the users are filled during the initial \textit{placement phase}, which takes place over a period of low traffic and high energy efficiency; and therefore, data delivery in the placement phase is assumed to be error-free and at a negligible energy cost\footnote{We assume that the placement phase takes place over a significantly longer period of time and over orthogonal high-quality links; which, in theory, allows the server to achieve the minimum energy per bit required to send the cache contents to the users.}; however, without either the knowledge of the user demands, or the users' future channel gains when they place their requests. The caching function for user $k \in [K]$ is
\begin{equation}\label{CachingFunction} 
{\phi _k}: {\left[ \left\lceil 2^{nR} \right\rceil \right]^N} \to \left[ \left\lfloor 2^{nMR} \right\rfloor \right],
\end{equation}
which maps the library to the cache contents $U_k$ of user $k$, i.e., $U_k={\phi _k}\left( \mathbf{W} \right)$, for $k\in [K]$.

During the peak traffic period, each user requests a single file from the library, where $d_k \in [N]$ denotes the index of the file requested by user $k \in [K]$. We assume that the user demands are independent and uniformly distributed over the file library; that is ${\rm{Pr}} \left\{ d_k=i \right\} = 1/N$, $\forall i \in [N], \forall k \in [K]$, and $\textbf{d} \buildrel \Delta \over = \left( d_1, ..., d_K \right)$ denotes the demand vector. The requests must be satisfied simultaneously during the \textit{delivery phase}. As opposed to the placement phase, the delivery phase takes place over a noisy BC, characterized by
\begin{equation}\label{ChannelModel} 
{Y_{k,i}} \left( \textbf{W},\textbf{d} \right) = h_{k} {X_{i}} \left( \textbf{W},\textbf{d} \right) + {Z_{k,i}}, \; \mbox{for $i \in [n], k \in [K]$},
\end{equation}
where ${X_{i}} \left( \textbf{W},\textbf{d} \right)$ denotes the transmitted signal from the server at time $i$, $h_{k}$ is a real channel gain between the server and user $k$, $Z_{k,i}$ is the zero-mean unit-variance real Gaussian noise at user $k$ at time $i$, i.e., $Z_{k,i} \sim \mathcal{N} \left( 0,1 \right)$, and ${Y_{k,i}} \left( \textbf{W},\textbf{d} \right)$ is the signal received at user $k$. The noise components are assumed independent and identically distributed (i.i.d.) across time and users. Without loss of generality, we assume that $h_1^2 \le h_2^2 \le \cdots \le h_K^2$. The channel input vector ${X^{n}} \left( \textbf{W},\textbf{d} \right)$ is generated by 
\begin{equation}\label{DeliveryFunction} 
\psi :{\left[ \left\lceil 2^{nR} \right\rceil \right]^N} \times {\left[ N \right]^K} \to \mathbb{R}^n,
\end{equation}     
and its average power is given by $P\left(\textbf{W},\textbf{d}\right) \buildrel \Delta \over = \frac{1}{n}\sum\nolimits_{i = 1}^n {X_{i}^2\left( \textbf{W},\textbf{d} \right)}$. We define the average power of this encoding function for demand vector $\textbf{d}$ as 
\begin{equation}\label{DefPowerConstraint}
P\left( \textbf{d} \right) \buildrel \Delta \over = \mathop {\max }\limits_{{W_1},...,{W_N}} P\left( {\textbf{W},\textbf{d}} \right),
\end{equation}
where the maximization is over all possible realizations of the file library. 

User $k \in [K]$ reconstructs ${\hat W}_{d_k}$ using its channel output $Y_{k}^n \left( \textbf{W},\textbf{d} \right)$, local cache contents $U_k$, channel vector $\textbf{h} \buildrel \Delta \over = \left(h_1, ..., h_K\right)$, and demand vector $\textbf{d}$ through the function
\begin{equation}\label{DecodingFunctionWeak} 
{\mu _k}:\mathbb{R}^n \times \left[ \left\lfloor 2^{nMR} \right\rfloor \right] \times \mathbb{R}^K \times \left[ N \right]^K \to \left[ \left\lceil 2^{nR} \right\rceil \right],
\end{equation}
where ${{\hat W}_{{d_k}}} = {\mu _k}\left( {Y_{k}^n \left( \textbf{W},\textbf{d} \right),U_k,\textbf{h},\textbf{d}} \right)$. The probability of error is defined as
\begin{equation}\label{ErrorProbability} 
{P_{e}} \buildrel \Delta \over = \Pr \left\{ \bigcup\nolimits_{\textbf{d} \in {[N]}^K}{{\bigcup\nolimits_{k = 1}^K {\left\{ {{{\hat W}_{{d_k}}} \ne {W_{{d_k}}}} \right\}} }} \right\}.
\end{equation}

An $(n,R,M)$ \textit{code} consists of $K$ caching functions ${\phi _1}, \dots, {\phi _K}$, channel encoding function $\psi$, and $K$ decoding functions ${\mu _1}, \dots, {\mu _K}$. We say that an $\left( R,M,\bar P, \hat P \right)$ tuple is \textit{achievable} if for every $\varepsilon > 0$, there exists an $(n,R,M)$ code with sufficiently large $n$, which satisfies ${P_{e}} < \varepsilon$, ${{\rm E}_{\textbf{d}}}\left[ {P(\textbf{d})} \right] \le \bar P$, and ${P(\textbf{d})} \le \hat P$, $\forall \textbf{d}$. For given rate $R$ and normalized cache capacity $M$, the average and peak power-memory trade-offs are defined, respectively, as
\begin{subequations}
\label{AverageandPeakPowerMemoryTradeOff}
\begin{align}\label{AveragePowerMemoryTradeOff}
\bar P^*\left( {R,M} \right) \buildrel \Delta \over = &\inf \left\{ {\bar P:\left( {R,M,\bar P,\infty } \right) \mbox{is achievable}} \right\},\\
{{\hat P}^*}\left( {R,M} \right) \buildrel \Delta \over =& \inf \left\{ {\hat P:\left( {R,M,\hat P,\hat P} \right) \mbox{is achievable}} \right\}.\label{PeakPowerMemoryTradeOff}
\end{align}
\end{subequations}

\begin{remark}
In the above definition, ${{\bar P}^*}$ is evaluated by allowing a different transmission power for each demand combination, and minimizing the average power across demands; while ${{\hat P}^*}$ characterizes the worst-case transmit power, which can also be considered as the minimum transmit power required to satisfy all possible demands if the transmitter is not allowed to adapt its transmit power according to user demands.
\end{remark}

%
%\begin{remark}
%In the model under consideration, we assume that the power required to perform the placement phase is much smaller than that of the delivery phase. The rationale behind this assumption is the abundance of the resources during the placement phase performed over off-peak traffic periods; furthermore, placing the contents into the users' caches can be performed when the broadcast channel between the server prefetching the contents and the users is relatively good, e.g., when the users are relatively close to the server. Thus, we aim to analyze the system performance in terms of the transmitted power during the delivery phase.
%\end{remark}
% 

We will consider both centralized and decentralized caching, and assume, in both scenarios that the placement phase is performed without any information about the channel gains during the delivery phase.

\begin{prop}\label{GeneralPowerAllocationLemma}
\textit{\cite{BergmansCapacityDegradeBC}} Consider a $K$-user Gaussian BC presented above with $M=0$, where a distinct message of rate $R_{k}$ is targeted for user $k$, $k \in [K]$. The minimum total power $P$ that is required to deliver all $K$ messages reliably can be achieved by superposition coding with Gaussian codewords of power $\alpha_kP$ allocated for user $k$, $\forall k \in [K]$, where\footnote{For two integers $i$ and $j$, if $i > j$, we assume that $\sum\nolimits_{n=i}^j {a_n}  = 0$, and $\prod\nolimits_{n=i}^j a_n  = 1$, where $a_n$ is an arbitrary sequence.}
\begin{align}\label{alphaiP}
&\alpha_{k}P = \left( \frac{2^{2R_{k}} - 1}{h_k^2} \right) \nonumber\\
& \;\;\;\;\;\quad \left( {1 + h_k^2 \sum\nolimits_{i = k + 1}^K \left( \frac{2^{2R_{i}} - 1}{h_i^2} \right)\prod\nolimits_{j = k + 1}^{i - 1} {{2^{2{R_{j}}}}} } \right),
\end{align}
and the total transmitted power is
\begin{equation}\label{PowerSetGeneralTerm}
{P} = \sum\nolimits_{k = 1}^{K} \left( \frac{2^{2R_{k}} - 1}{h_k^2} \right)\prod\nolimits_{i = 1}^{k - 1} {{2^{2{R_{i}}}}} .
\end{equation}
\end{prop}

For a demand vector $\textbf{d}$ in the delivery phase, we denote the number of distinct demands by $N_{\textbf{d}}$, where $N_{\textbf{d}} \le \min \left\{ {N,K} \right\}$. Let $\mathcal U_{\textbf{d}}$ denote the set of users with distinct requests, which have the worst channel qualities; that is, $\mathcal{U}_{\textbf{d}}$ consists of $N_{\textbf{d}}$ indices corresponding to users with distinct requests, where a user is included in set $\mathcal{U}_{\textbf{d}}$ iff it has the worst channel quality among all the users with the same demand, i.e., if $k \in \mathcal{U}_{\textbf{d}}$ and $d_k = d_m$ for some $m \in [K]$, then $h_k^2 \le h_{m}^2$, or equivalently, $k \le m$. Note that, for any demand vector $\textbf{d}$, $1 \in \cal{U}_{\textbf{d}}$. For given $\textbf{d}$ and user $k$, $k \in [K]$, let $\mathcal{U}_{\textbf{d},k}$ denote the set of users in $\mathcal{U}_{\textbf{d}}$ which have better channels than user $k$:
\begin{equation}\label{DefSetUdi}
{\cal U}_{\textbf{d},k} \buildrel \Delta \over = \left\{ i \in \mathcal {U}_{\textbf{d}}:\mbox{$i > k$} \right\}, \quad \mbox{$k \in [K]$}.
\end{equation}
We denote the cardinality of $\mathcal{U}_{\textbf{d},k}$ by $N_{\textbf{d},k}$, i.e., $N_{\textbf{d},k} \buildrel \Delta \over = \left| \mathcal{U}_{\textbf{d},k} \right|$.

\section{Centralized Caching and Delivery}\label{ProposedSchemeCentralized} 

Here we present our centralized caching and coded delivery scheme. We follow the placement phase in \cite{YuMaddahAliAvestimehrExact} since the users' channel gains are not known in advance. In the delivery phase, our goal is to identify the coded packets targeted to each user in order to minimize the transmit power. We will use superposition coding in sending multiple coded packets, and benefit heavily from the degradedness of the underlying Gaussian BC.

\begin{theorem}\label{UpperboundPowerMemoryTheoremCentralized}
In centralized caching followed by delivery over a Gaussian BC, we have
\begin{subequations}
\label{AchievablePowerMemoryTheoremDemandCentralized}
\begin{align}\label{AchievablePowerMemoryTheoremDemandAchievablePowerMemoryTheoremWorstCaseTheoremCentralized1}
& \bar P^*(R,M) \le \bar  P_{\rm{UB}}^{\rm{C}}(R,M) \buildrel \Delta \over = \nonumber\\
&\qquad \quad \frac{1}{N^K} \sum\limits_{\emph{\textbf{d}} \in \left[ N \right]^K} { \left[ \sum\limits_{i = 1}^K {\left( \frac{{{2^{2{R^{\rm{C}}_{\emph{\textbf{d}},i}}}} - 1}}{h_i^2} \right)\prod\limits_{j = 1}^{i - 1} {{2^{2{R^{\rm{C}}_{\emph{\textbf{d}},j}}}}} } \right] },
\end{align}
where
%\footnote{Given two integers $i$ and $j$ with $i < j$, we assume $\binom{i}{j}=0$.}
\begin{align}\label{AchievablePowerMemoryTheoremDemandAchievablePowerMemoryTheoremWorstCaseTheoremCentralized2}
{R^{\rm{C}}_{\emph{\textbf{d}},k}} \buildrel \Delta \over = 
\begin{cases} 
\frac{\binom{K-k}{\left\lfloor t \right\rfloor}}{\binom{K}{\left\lfloor t \right\rfloor}}\left( \left\lfloor t \right\rfloor +1 -t \right)R+\\
\qquad \frac{\binom{K-k}{\left\lfloor t \right\rfloor+1}}{\binom{K}{\left\lfloor t \right\rfloor+1}}\left(t- \left\lfloor t \right\rfloor \right)R, &\mbox{if $k \in \mathcal{U}_{\emph{\textbf{d}}}$},\\
\frac{\binom{K-k}{\left\lfloor t \right\rfloor}-\binom{K-k-N_{\emph{\textbf{d}},k}}{\left\lfloor t \right\rfloor}}{\binom{K}{\left\lfloor t \right\rfloor}}\left( \left\lfloor t \right\rfloor +1 -t \right)R + \\
\qquad \frac{\binom{K-k}{\left\lfloor t \right\rfloor+1}-\binom{K-k-N_{\emph{\textbf{d}},k}}{\left\lfloor t \right\rfloor+1}}{\binom{K}{\left\lfloor t \right\rfloor+1}}\left(t- \left\lfloor t \right\rfloor \right)R, &\mbox{otherwise},
\end{cases}
\end{align}
\end{subequations}
and 
\begin{subequations}
\label{AchievablePowerMemoryTheoremWorstCaseCentralized}
\begin{align}\label{AchievablePowerMemoryTheoremWorstCaseCentralized1}
\hat P^*(R,M) \le \hat P^{\rm{C}}_{\rm{UB}}(R,M) \buildrel \Delta \over = \sum\limits_{i = 1}^{K} {\left( \frac{{2^{2{\hat R^{\rm{C}}_{\emph{\textbf{d}},i}}} - 1}}{h_i^2} \right)\prod\limits_{j = 1}^{i - 1} {{2^{2{\hat R^{\rm{C}}_{\emph{\textbf{d}},j}}}}} },
\end{align}
where 
\begin{align}\label{AchievablePowerMemoryTheoremWorstCaseCentralized2}
{\hat R^{\rm{C}}_{\emph{\textbf{d}},k}} \buildrel \Delta \over = 
\begin{cases} 
\frac{\binom{K-k}{\left\lfloor t \right\rfloor}}{\binom{K}{\left\lfloor t \right\rfloor}}\left( \left\lfloor t \right\rfloor +1 -t \right)R+\\
\qquad \frac{\binom{K-k}{\left\lfloor t \right\rfloor+1}}{\binom{K}{\left\lfloor t \right\rfloor+1}}\left(t- \left\lfloor t \right\rfloor \right)R, &\mbox{if $k \in \left[ \min\{N,K\} \right]$},\\
0, &\mbox{otherwise}.
\end{cases}
\end{align}
\end{subequations}
\end{theorem}

\begin{proof}
For simplicity, we assume that both $nR$ and $nMR$ are integers. The proposed scheme is first presented for integer normalized global cache capacities, that is $t \in [0:K]$. The scheme is then extended to any $t \in [0,K]$.   

\subsection{Integer $t$ Values}\label{CentralizedIntegerValues}

Here, we assume that $t \in [0:K]$. We denote the $t$-element subsets of $K$ users by $\mathcal{C}_1^{t}, \mathcal{C}_2^{t}, \dots, \mathcal{C}_{\binom{K}{t}}^{t}$. 

%\subsection{Placement Phase}\label{PlacementCentralized}
\underline{\textbf{Placement phase:}} A centralized cache placement phase is performed without the knowledge of the future user demands or the channel gains in the delivery phase. Each file $W_i$, $i \in [N]$, is split into $\binom{K}{t}$ equal-length subfiles $W_{i,\mathcal{C}_1^{t}}, W_{i,\mathcal{C}_2^{t}},$ $\dots, W_{i,\mathcal{C}_{\binom{K}{{t}}}^{t}}$, each of rate $R/\binom{K}{{t}}$. User $k$, $k \in [K]$, caches subfile $W_{i,\mathcal{C}_l^{t}}$, if $k \in \mathcal{C}_l^{t}$, for $i \in [N]$ and $l \in \left[\binom{K}{{t}}\right]$. Hence, the cache contents of user $k$ is given by
\begin{equation}\label{CacheContentUserkDecentralized}
{U_k} = \bigcup\nolimits_{i \in \left[ N \right]} {\bigcup\nolimits_{l \in \left[\binom{K}{{t}}\right]: k \in \mathcal{C}_l^{t}} {W_{i,\mathcal{C}_l^{t}}} }, \quad \mbox{for $k \in [K]$},
\end{equation}
where the cache capacity constraint is satisfied with equality. 

%\subsection{Delivery Phase}\label{DeliveryCentralized}

\underline{\textbf{Delivery phase:}} For an arbitrary $\textbf{d}$, we will deliver the following coded message to the users in $\mathcal{C}_l^{t+1}$, $\forall l \in \left[ \binom{K}{t+1} \right]$:
\begin{equation}\label{DefCodedDeliveredContentCentralized}
Q_{\mathcal{C}_l^{{t}+1}} \buildrel \Delta \over = {{\bigoplus}_{k \in \mathcal{C}_l^{{t}+1}} {W_{{d_k},\mathcal{C}_l^{{t}+1}\backslash \{ k\} }}}, 
\end{equation}
where ${\oplus}$ represents the bitwise XOR operation. Then, each user $i \in \mathcal{C}_l^{{t}+1}$ can recover subfile $W_{{d_i},\mathcal{C}_l^{{t}+1}\backslash \{ i\} }$, since it has cached all the other subfiles $W_{{d_j},\mathcal{C}_l^{{t}+1}\backslash \{ j\} }$, $\forall j \in \mathcal{C}_l^{{t}+1} \backslash \{i\}$. Note that each coded message $Q_{\mathcal{C}_l^{{t}+1}}$, $l\in \left[ \binom{K}{{t}+1} \right]$, is of rate $R/\binom{K}{{t}}$. Note also that, for $k \in [K]$, sending $\bigcup\nolimits_{l:k \in \mathcal {C}_l^{{t} + 1}} Q_{\mathcal{C}_l^{{t}+1}}$ to user $k$ enables that user to obtain all the subfiles $W_{{d_k},\mathcal{C}_l^{{t}}}$, $\forall l \in \left[ \binom{K}{{t}} \right]$ and $k \notin {\mathcal{C}_l^{{t}}}$. Thus, together with its cache contents, the demand of user $k$, $k \in [K]$ would be fully satisfied after receiving $\bigcup\nolimits_{l:k \in \mathcal {C}_l^{{t} + 1}} Q_{\mathcal{C}_l^{t+1}}$. As observed in \cite{YuMaddahAliAvestimehrExact}, for a demand vector $\textbf{d}$ with $N_{\textbf{d}}$ distinct requests, if $K-N_{\textbf{d}}\ge {t}+1$, not all the coded messages $Q_{\mathcal{C}_l^{{t}+1}}$, $\forall l \in \left[ \binom{K}{{t}} \right]$, need to be delivered.

Following \cite[Lemma 1]{YuMaddahAliAvestimehrExact}, for a demand vector ${\textbf{d}}$ with $N_{{\textbf{d}}}$ distinct requests, let $\mathcal U_{{\textbf{d}}} \subset {\cal B} \subset [K]$. We define $\mathcal {G}_{\cal B}$ as the set consisting of all the subsets of $\cal B$ with cardinality $N_{{\textbf{d}}}$, such that all $N_{{\textbf{d}}}$ users in each subset request distinct files. For any $\cal B$, we have ${\bigoplus }_{\mathcal G \in {\mathcal{G}_{\mathcal B}}} {Q_{\mathcal B\backslash \mathcal G}} = \textbf{0}$, where \textbf{0} denotes the zero vector.

\begin{remark}\label{RemAfterLemma}
Given a demand vector $\emph{\textbf{d}}$ with $N_{\emph{\textbf{d}}} < K$, and any set ${\cal S} \subset [K]\backslash \mathcal U_{\emph{\textbf{d}}}$ of users, by setting $\mathcal{B} = {\cal S} \cup \cal {U}_{\emph{\textbf{d}}}$, we have
\begin{align}\label{LemmaEqConc}
{\bigoplus }_{{\cal G} \in {{\cal G}_{\cal B}}} {Q_{{\cal B}\backslash {\cal G}}} & = \left( {\bigoplus }_{{\cal G} \in {{\cal G}_{\cal B}}\backslash \mathcal U_{\emph{\textbf{d}}}} {Q_{{\cal B}\backslash {\cal G}}} \right) \oplus {Q_{{\cal B}\backslash \mathcal U_{\emph{\textbf{d}}}}} \nonumber\\
& = \left( {\bigoplus }_{{\cal G} \in {{\cal G}_{\cal B}}\backslash \mathcal U_{\emph{\textbf{d}}}} {Q_{{\cal B}\backslash {\cal G}}} \right) \oplus {Q_{\mathcal {S}}} = \emph{\textbf{0}},
\end{align}
which leads to
\begin{equation}\label{LemmaEqConcLast}
{Q_{\mathcal {S}}} = {\bigoplus }_{\mathcal {G} \in {\mathcal {G}_{\cal B}}\backslash {\mathcal U_{\emph{\textbf{d}}}}} {Q_{\mathcal{B}\backslash \mathcal {G}}}.
\end{equation}
Thus, having received all the coded messages $Q_{\mathcal{B} \backslash \mathcal{G}}$, $\forall \mathcal {G} \in {\mathcal {G}_{\cal B}}\backslash {\mathcal U_{\emph{\textbf{d}}}}$, $Q_{\mathcal {S}}$ can be recovered through \eqref{LemmaEqConcLast}. Note that, for any $\mathcal {G} \in {\mathcal {G}_{\cal B}}\backslash {\cal U}_{\emph{\textbf{d}}}$, we have 
\begin{subequations}
\label{ResultsLemma}
\begin{align}\label{ResultsLemma1}
&\left| {{\cal B}\backslash {\cal G}} \right| = \left| \mathcal S \right|,\\
&\left( {\cal B}\backslash {\cal G} \right) \cap {\cal U}_{\emph{\textbf{d}}} \ne \emptyset,
\label{ResultsLemma2}
\end{align}
\end{subequations}
that is, each coded message on the right hand side (RHS) of \eqref{LemmaEqConcLast} is targeted for a set of $\left| \mathcal S \right|$ users, at least one of which is in set $\mathcal U_{\emph{\textbf{d}}}$. Furthermore, for each $k \in \cal S$, there is a user $k' \in \mathcal U_{\emph{\textbf{d}}}$ with $h_{k'}^2 \le h_{k}^2$, such that $d_{k'} = d_k$. Note that, since no two users with the same demand are in any of the sets ${\cal G} \in {\mathcal {G}_{\cal B}}$, for any set $\mathcal {G} \in {\mathcal {G}_{\cal B}}\backslash {\cal U}_{\emph{\textbf{d}}}$, we have either $k \in {\cal B} \backslash \mathcal {G}$ or $k' \in {\cal B} \backslash \mathcal {G}$.
\end{remark}

Given a demand vector $\textbf{d}$, the delivery phase is designed such that only the coded messages $Q_{\mathcal{C}_l^{{t}+1}}$, $\forall l \in \left[ \binom{K}{{t}+1} \right]$ such that ${\mathcal{C}_l^{{t}+1}} \cap {\cal U}_{{\textbf{d}}} \ne \emptyset$, are delivered, i.e., the coded messages that are targeted for at least one user in ${\cal U}_{{\textbf{d}}}$ are delivered, and the remaining coded messages can be recovered through \eqref{LemmaEqConcLast}. To achieve this, for any such set ${\mathcal{C}_l^{{t}+1}}$ with ${\mathcal{C}_l^{{t}+1}} \cap {\cal U}_{{\textbf{d}}} \ne \emptyset$, the transmission power is adjusted such that the worst user in ${\mathcal{C}_l^{{t}+1}}$ can decode $Q_{\mathcal{C}_l^{{t}+1}}$; and so can all the other users in ${\mathcal{C}_l^{{t}+1}}$ due to the degradedness of the Gaussian BC. As a result, the demand of every user in $\cal{U}_{\textbf{d}}$ will be satisfied.

We aim to find the coded packets targeted for each user that will minimize the transmitted power, while guaranteeing that all the user demands are satisfied. In delivering the coded messages, we start from the worst user, i.e., user $1$, and first transmit all the coded messages targeted for user 1. We then target the second worst user, and transmit the coded messages targeted for it that have not been already delivered, keeping in mind that only the coded messages $Q_{\mathcal{C}_l^{{t}+1}}$ for which ${\mathcal{C}_l^{{t}+1}} \cap {\cal U}_{{\textbf{d}}} \ne \emptyset$ are delivered. We continue similarly until we deliver the messages targeted for the best user in ${\cal U}_{{\textbf{d}}}$. 

We denote the contents targeted for user $k$ by ${\tilde Q}_k$, and their total rate by $R_{\textbf{d},{t},k}$, $k \in [K]$. For a demand vector $\textbf{d}$, contents
\begin{equation}\label{CentIntMesUser1}
{\tilde Q}_1 = \bigcup\nolimits_{l:1 \in \mathcal {C}_l^{{t} + 1}} Q_{\mathcal{C}_l^{{t}+1}}
\end{equation}
are targeted for user 1. Note that there are $\binom{K-1}{{t}}$ different $(t+1)$-element subsets $\mathcal {C}_l^{{t} + 1}$, in which user 1 is included. Thus, the total rate of the messages targeted for user 1 is
\begin{equation}\label{TotalRateIntendedUserOneCentralized}
R_{\textbf{d},{t},1} = \frac{\binom{K-1}{{t}}}{\binom{K}{{t}}}R = \frac{K-t}{K}R. 
\end{equation}
For $k \in [2:K]$, the coded contents targeted for user $k$ that have not been sent through the transmissions to the previous $k-1$ users and are targeted for at least one user in ${\cal U}_{{\textbf{d}}}$ are delivered. Thus, for $k \in [2:K]$, we deliver
\begin{equation}\label{CodedContentsUserkNotinLeadersCentralized}
{\tilde Q}_k = \bigcup\nolimits_{l: {\mathcal{C}_l^{t+1}} \cap {\cal U}_{{\textbf{d}}} \ne \emptyset, {\mathcal{C}_l^{t+1}} \cap [k-1] = \emptyset, k \in {\mathcal{C}_l^{t+1}}} Q_{\mathcal{C}_l^{t+1}}, 
\end{equation}
which is equivalent to
\begin{align}\label{CodedContentsUserkNotinLeadersEquivalentCentralized}
{\tilde Q}_k = & \left(\bigcup\nolimits_{l: {\mathcal{C}_l^{t+1}} \cap [k-1] = \emptyset, k \in {\mathcal{C}_l^{t+1}}} Q_{\mathcal{C}_l^{t+1}}\right) - \nonumber\\
& \left( \bigcup\nolimits_{l: {\mathcal{C}_l^{t+1}} \cap {\cal U}_{{\textbf{d}}} = \emptyset, {\mathcal{C}_l^{t+1}} \cap [k-1] = \emptyset, k \in {\mathcal{C}_l^{t+1}}} Q_{\mathcal{C}_l^{t+1}} \right). 
\end{align}
For each user $k \in [2:K]$, there are $\binom{K-k}{t}$ different $(t+1)$-element subsets ${\mathcal{C}_l^{t+1}}$, such that ${\mathcal{C}_l^{t+1}} \cap [k-1] = \emptyset$ and $k \in {\mathcal{C}_l^{t+1}}$, for $l=1, \dots, \binom{K}{t+1}$. On the other hand, for each user $k \in [2:K] \backslash \mathcal{U}_{\textbf{d}}$, since there are $N_{{\textbf{d}},k}$ users among the set of users $[k:K]$ that belong to set $\cal{U}_{\textbf{d}}$, there are $\binom{K-k-N_{{\textbf{d}},k}}{t}$ different $(t+1)$-element subsets ${\mathcal{C}_l^{t+1}}$, such that ${\mathcal{C}_l^{t+1}} \cap \left( {\cal{U}}_{\textbf{d}} \cup [k-1] \right)= \emptyset$ and $k \in {\mathcal{C}_l^{t+1}}$, for $l=1, \dots, \binom{K}{t+1}$. Note that, if $k \in {\cal{U}}_{\textbf{d}}$, the second term on the RHS of \eqref{CodedContentsUserkNotinLeadersEquivalentCentralized} includes no content. Thus, if $k \in [2:K] \cap \cal{U}_{\textbf{d}}$, total rate targeted for user $k$ is
\begin{equation}\label{TotalRateIntendedUserkinLeadersCentralized}
R_{\textbf{d},t,k} = \frac{\binom{K-k}{t}}{\binom{K}{t}}R = \left( {\prod\nolimits_{i = 0}^{k - 1} {\frac{{K - t - i}}{{K - i}}} } \right)R, 
\end{equation} 
while the total rate targeted for user $k$, $k \in [2:K] \backslash \cal{U}_{\textbf{d}}$, is
\begin{align}\label{TotalRateIntendedUserkNotinLeadersCentralized}
R_{\textbf{d},t,k} =\frac{\binom{K-k}{t}-\binom{K-k-N_{{\textbf{d}},k}}{t+1}}{\binom{K}{t}}R.
\end{align}
In summary, the proposed achievable scheme intends to deliver contents of total rate $R_{\textbf{d},t,k}$ to user $k$, for $k \in [K]$, where
\begin{align}\label{TotalRateEachUserCentralized}
{R_{{\textbf{d}},t,k}} \buildrel \Delta \over = 
\begin{cases} 
\frac{\binom{K-k}{t}}{\binom{K}{t}}R, &\mbox{if $k \in \mathcal{U}_{{\textbf{d}}}$},\\
\frac{\binom{K-k}{t}-\binom{K-k-N_{{\textbf{d}},k}}{t}}{\binom{K}{t}}R, &\mbox{otherwise}.
\end{cases}
\end{align}
The centralized caching and coded packet generation explained above is summarized in Algorithm \ref{CentralizedSchemeAlg}.

\begin{algorithm}[t]
\caption{Centralized Caching and Coded Packet Generation}
\label{CentralizedSchemeAlg}
\begin{algorithmic}[1]
\Statex
\Procedure {Placement Phase}{}
\State{$W_i = {\bigcup\nolimits_{l =1}^{\binom{K}{{t}}} {W_{i,\mathcal{C}_l^{t}}} }$, \quad for $i=1, ..., N$}
\Statex
\State{${U_k} = \bigcup\nolimits_{i \in \left[ N \right]} {\bigcup\nolimits_{l \in \left[\binom{K}{{t}}\right]: k \in \mathcal{C}_l^{t}} {W_{i,\mathcal{C}_l^{t}}} }, \quad \mbox{for $k=1, \dots, K$}$}
\EndProcedure
\Statex
\Procedure {Coded Packet Generation}{}
\State{$Q_{\mathcal{C}_l^{{t}+1}} = {{\bigoplus}_{k \in \mathcal{C}_l^{{t}+1}} {W_{{d_k},\mathcal{C}_l^{{t}+1}\backslash \{ k\} }}}$, \; for $l=1, ..., \binom{K}{t+1}$}
\Statex
\State{${\tilde Q}_k = \bigcup\nolimits_{l: {\mathcal{C}_l^{t+1}} \cap {\cal U}_{{\textbf{d}}} \ne \emptyset, {\mathcal{C}_l^{t+1}} \cap [k-1] = \emptyset, k \in {\mathcal{C}_l^{t+1}}} Q_{\mathcal{C}_l^{t+1}}$, \; for $k=1, \dots, K$}
\EndProcedure
\end{algorithmic}
\end{algorithm}

Once the coded packets targeted to each user are determined, the next step is to design the physical layer coding scheme to deliver these packets over the Gaussian BC. Given $\textbf{d}$, we generate $K$ codebooks, where the codebook $k \in [K]$ is designed to deliver the contents $\tilde{Q}_k$ to user $k$. The $k$-th codebook consists of $2^{n{R_{{\textbf{d}},t,k}}}$ i.i.d. Gaussian codewords $x_{k}^n \left( \textbf{W},\textbf{d} \right)$, generated according to the normal distribution $\mathcal{N} \left( 0,\alpha_k P^{\rm{C}}_{{\rm{UB}}}\left( R,M,\textbf{d} \right) \right)$, where $\alpha_k \ge 0$ and $\sum\nolimits_{i = 1}^K {{\alpha _i}}  = 1$. The transmission is performed through superposition coding; that is, the channel input is given by $\sum\nolimits_{k = 1}^K {x_{k}^n \left(\textbf{W},\textbf{d}, {\tilde q}_k \right)}$, where ${\tilde q}_k \in \left[ 2^{n{R^{\rm{C}}_{{\textbf{d}},t,k}}} \right]$. User $k$ decodes codewords $x_{1}^n \left(\textbf{W},\textbf{d} \right), ..., x_{k}^n \left(\textbf{W},\textbf{d} \right)$ by successive decoding, considering all the codewords in higher levels as noise. If all the $k$ codewords are decoded successfully, user $k \in [K]$ can recover contents ${\tilde Q}_1, ..., {\tilde Q}_k$. Accordingly, any user $k \in [K]$ will receive all the coded contents targeted for it except those that are not intended for at least one user in ${\cal U}_{{\textbf{d}}}$ (which have not been delivered); that is, user $k$ receives all the coded contents
\begin{equation}\label{CodedContentsCentIntUserkContentk}
\bigcup\nolimits_{l: {\mathcal{C}_l^{t+1}} \cap {\cal U}_{{\textbf{d}}} \ne \emptyset, k \in {\mathcal{C}_l^{t+1}}} Q_{\mathcal{C}_l^{t+1}}.
\end{equation}
Thanks to the degradedness of the underlying Gaussian BC, it can also obtain all the coded contents targeted for users in $[k-1]$, which are also intended for at least one user in ${\cal U}_{{\textbf{d}}}$, i.e., all the coded contents in
\begin{equation}\label{CodedContentsCentIntUserkContent1_k}
\bigcup\nolimits_{l: {\mathcal{C}_l^{t+1}} \cap {\cal U}_{{\textbf{d}}} \ne \emptyset, [k-1] \cap {\mathcal{C}_l^{t+1}}\ne \emptyset} Q_{\mathcal{C}_l^{t+1}}.
\end{equation}
Note that, if $k \in {\cal U}_{{\textbf{d}}}$, \eqref{CodedContentsCentIntUserkContentk} reduces to $\bigcup\nolimits_{l:k \in {\mathcal{C}_l^{t+1}}} Q_{\mathcal{C}_l^{t+1}}$, which shows that the demand of user $k \in {\cal U}_{{\textbf{d}}}$ is satisfied. Next, we illustrate that the users in $[K] \backslash {\cal U}_{{\textbf{d}}}$ can obtain their requests without being delivered any extra messages. Given any set of users ${\mathcal{C}_l^{t+1}}$ such that ${\mathcal{C}_l^{t+1}} \cap {\cal U}_{{\textbf{d}}} = \emptyset$, we need to show that every user in ${\mathcal{C}_l^{t+1}}$ can decode all the coded messages $Q_{\mathcal B \backslash \cal G}$, $\forall \mathcal G \in {\mathcal {G}_{\cal B}}\backslash {\mathcal U_{{\textbf{d}}}}$, where $\mathcal B ={\mathcal{C}_l^{t+1}} \cup {\cal U}_{{\textbf{d}}}$. In this case, they can also decode $Q_{\mathcal{C}_l^{t+1}}$ through \eqref{LemmaEqConcLast}. 

Assume that there exists a set of users ${\mathcal{C}_l^{t+1}}$ such that ${\mathcal{C}_l^{t+1}} \cap {\cal U}_{{\textbf{d}}} = \emptyset$; set $\mathcal B ={\mathcal{C}_l^{t+1}} \cup {\cal U}_{{\textbf{d}}}$. According to \eqref{ResultsLemma2}, there is at least one user in $\mathcal{U}_\textbf{d}$ in any set of users $\mathcal B \backslash \cal G$, $\forall \mathcal G \in {\mathcal {G}_{\cal B}}\backslash {\mathcal U_{{\textbf{d}}}}$. Thus, all the coded messages $Q_{\mathcal B \backslash \cal G}$ have been delivered by the proposed delivery scheme. Remember that, for each user $k \in {\mathcal{C}_l^{t+1}}$ and $\forall \mathcal G \in {\mathcal {G}_{\cal B}}\backslash {\mathcal U_{{\textbf{d}}}}$, either $k \in \mathcal B \backslash \cal G$ or $\exists k' \in \mathcal B \backslash \cal G$, where $d_{k'}=d_k$ and $k' \in \mathcal{U}_\textbf{d}$, i.e., $h^2_{k'} \le h^2_{k}$. If $k \in \mathcal B \backslash \cal G$, since $\mathcal B \backslash \cal G \cap {\mathcal U_{{\textbf{d}}}} \ne \emptyset$, according to \eqref{CodedContentsCentIntUserkContentk}, user $k$ can obtain $Q_{\mathcal B \backslash \cal G}$. If $\exists k' \in \mathcal B \backslash \cal G$ with $d_{k'} = d_{k}$ and $k' \in {\mathcal U_{{\textbf{d}}}}$, then user $k'$ can decode $Q_{\mathcal B \backslash \cal G}$, and since $h^2_{k'} \le h^2_{k}$, user $k$ can also decode $Q_{\mathcal B \backslash \cal G}$. Thus, each user $k \in {\mathcal{C}_l^{t+1}}$ can decode $Q_{\mathcal{C}_l^{t+1}}$ successfully, for any set of users ${\mathcal{C}_l^{t+1}}$ that satisfies ${\mathcal{C}_l^{t+1}} \cap {\cal U}_{{\textbf{d}}} = \emptyset$. This fact confirms that the demands of all the users in $[K] \backslash {\cal U}_{{\textbf{d}}}$ are satisfied.

We remark here that, with the proposed coded delivery scheme, the coded packets targeted to user $k$ are delivered only to users in $[k]$, and not to the any of the users in $[k+1:K]$, for $k \in [K]$.
%Due to the degradedness of the Gaussian BC, we proved that this coded delivery scheme can satisfy any demand combination by utilizing superposition coding taking the same demands into account.   
Next, we provide an example to illustrate the proposed joint cache and channel coding scheme.

\underline{\textbf{Example:}} Consider $K=5$ users, each equipped with a cache of normalized size $M = N/5$, i.e., $t=1$. File $W_i$ is partitioned into 5 equal-length subfiles $W_{i,\{1 \}}$, $\ldots$, $W_{i,\{5 \}}$, each of which is of rate $R/5$, for $i \in [N]$. Cache contents of user $k$ after the placement phase is given by
\begin{align}\label{CacheContentsExmp}
{U_k} & = \bigcup\nolimits_{i \in \left[ N \right]} W_{i,\left\{ k \right\}}, \quad k \in [5].
\end{align}
Assuming that $N \ge 3$, let the user demands be $\textbf{d} = \left( 1,2,1,1,3 \right)$, where we have $N_{{\textbf{d}}} =3$, and $\mathcal U_{{\textbf{d}}} = \{ 1,2,5\}$. We generate the following coded packets:
\begin{subequations}
\label{CodedPacketGenExmp}
\begin{align}\label{CodedPacketGenExmp1}
Q_{\{1,2\}} & = W_{1,\{2 \}} \oplus W_{2,\{1 \}},\\
Q_{\{1,3\}} & = W_{1,\{3 \}} \oplus W_{1,\{1 \}} ,\label{CodedPacketGenExmp2}\\
Q_{\{1,4\}} & = W_{1,\{4 \}} \oplus W_{1,\{1 \}},
\label{CodedPacketGenExmp3}\\
Q_{\{1,5\}} & = W_{1,\{5 \}} \oplus W_{3,\{1 \}},
\label{CodedPacketGenExmp4}\\
Q_{\{2,3\}} & = W_{2,\{3 \}} \oplus W_{1,\{2 \}},
\label{CodedPacketGenExmp5}\\
Q_{\{2,4\}} & = W_{2,\{4 \}} \oplus W_{1,\{2 \}},
\label{CodedPacketGenExmp6}\\
Q_{\{2,5\}} & = W_{2,\{5 \}} \oplus W_{3,\{2 \}},
\label{CodedPacketGenExmp7}\\
Q_{\{3,5\}} & = W_{1,\{5 \}} \oplus W_{3,\{3 \}},
\label{CodedPacketGenExmp8}\\
Q_{\{4,5\}} & = W_{1,\{5 \}} \oplus W_{3,\{4 \}}.
\label{CodedPacketGenExmp9}
\end{align}
\end{subequations}
The coded contents ${\tilde Q}_k$ targeted to user $k$, $k \in [5]$, are:
\begin{subequations}
\label{QktildeExmp}
\begin{align}\label{QktildeExmp1}
{\tilde Q}_1 & = Q_{\{1,2\}}, Q_{\{1,3\}}, Q_{\{1,4\}}, Q_{\{1,5\}},\\
{\tilde Q}_2 & = Q_{\{2,3\}}, Q_{\{2,4\}}, Q_{\{2,5\}} ,\label{QktildeExmp2}
\end{align}

\begin{align}
{\tilde Q}_3 & = Q_{\{3,5\}},
\label{QktildeExmp3}\\
{\tilde Q}_4 & = Q_{\{4,5\}},
\label{QktildeExmp4}\\
{\tilde Q}_5 & = \textbf{0}.
\label{QktildeExmp5}
\end{align}
\end{subequations}
We perform superposition coding to deliver these contents, and send $\sum\nolimits_{k = 1}^4 {x_{k}^n \left(\textbf{W},\textbf{d}, {\tilde q}_k \right)}$, where ${\tilde q}_1 \in \left[ 2^{4nR/5} \right]$, ${\tilde q}_2 \in \left[ 2^{3nR/5} \right]$, ${\tilde q}_3 \in \left[ 2^{nR/5} \right]$, and ${\tilde q}_4 \in \left[ 2^{nR/5} \right]$. User $k$, $k \in [5]$, decodes the codewords $x_{1}^n \left(\textbf{W},\textbf{d} \right), ..., x_{k}^n \left(\textbf{W},\textbf{d} \right)$ by successive decoding, considering the higher level codewords as noise. If successful, user $k$ can recover ${\tilde Q}_1, \ldots, {\tilde Q}_k$, $k \in [5]$. Now, note that user 1 can recover $W_1$ from ${\tilde Q}_1$ together with its cache contents; user 2 can decode $W_2$ having received the coded packets $Q_{\{1,2\}}$, $Q_{\{2,3\}}$, $Q_{\{2,4\}}$, and $Q_{\{2,5\}}$ along with its cache contents; coded packets $Q_{\{1,5\}}$, $Q_{\{2,5\}}$, $Q_{\{3,5\}}$, and $Q_{\{4,5\}}$, and its cache contents $U_5$ enable user 5 to decode $W_3$; user 3 can directly obtain $Q_{\{1,3\}}$, $Q_{\{2,3\}}$, and $Q_{\{3,5\}}$, and it can generate $Q_{\{3,4\}}$ as follows:
\begin{align}\label{Q34Exmpl}
Q_{\{3,4\}} = Q_{\{1,3\}} \oplus Q_{\{1,4\}},    
\end{align}
so it can decode $W_1$ together with its cache contents. Note that both $Q_{\{1,3\}}$ and $Q_{\{1,4\}}$ are delivered within ${\tilde Q}_1$, which is decoded by user 3. Similarly, user 4 will be delivered $Q_{\{1,4\}}$, $Q_{\{2,4\}}$, and $Q_{\{4,5\}}$, and it can also recover $Q_{\{3,4\}}$ through \eqref{Q34Exmpl}, and together with its cache contents it can recover $W_1$.   
\qed

\subsection{Non-integer $t$ Values}\label{CentralizedGeneralValues}
Here we extend the proposed scheme to non-integer $t$ values. We divide the whole database as well as the cache memory of the users into two, such that the corresponding $t$ parameters for both parts are integer. This way we can employ the placement and delivery schemes introduced in Section \ref{CentralizedIntegerValues} for each part separately. 

\underline{\textbf{Placement phase:}} Each file $W_i$, for $i=1, ..., N$, is divided into two non-overlapping subfiles, $W_i^1$ of rate $R_1$ and $W_i^2$ of rate $R_2$. We set
\begin{subequations}
\label{NonintegerRateEachSubfile}
\begin{align}\label{NonintegerRateEachSubfile1}
R_1 &= \left( \left\lfloor t \right\rfloor +1-t \right)R,\\
R_2 &= \left( t - \left\lfloor t \right\rfloor \right)R,\label{NonintegerRateEachSubfile2}
\end{align}
\end{subequations}
such that $R_1+R_2=R$. For subfiles $\left\{ W_1^1, \dots, W_N^1 \right\}$, we perform the placement phase proposed in Section \ref{CentralizedIntegerValues} corresponding to the normalized global cache capacity $\left\lfloor t \right\rfloor$, which requires a cache capacity of $n \left( \left\lfloor t \right\rfloor N/K \right)R_1 \mbox{ bits}$. While, for subfiles $\left\{ W_1^2, \dots, W_N^2 \right\}$, we perform the placement phase proposed in Section \ref{CentralizedIntegerValues} corresponding to the normalized global cache capacity $\left\lfloor t \right\rfloor +1$, which requires a cache capacity of $n \left( \left( \left\lfloor t \right\rfloor + 1 \right) N/K \right)R_2 \mbox{ bits}$. By summing up, the total cache capacity is found to be $nMR$ bits, which shows that the cache capacity constraint is satisfied.

\underline{\textbf{Delivery phase:}} For any demand vector $\textbf{d}$, we perform the proposed coded delivery phase in Section \ref{CentralizedIntegerValues} corresponding to the normalized global cache capacity $\left\lfloor t \right\rfloor$ to deliver the missing bits of subfiles $\left\{ W_1^1, \dots, W_N^1 \right\}$ to the intended users. Moreover, the missing bits of subfiles $\left\{ W_1^2, \dots, W_N^2 \right\}$ are delivered to the intended users by performing the coded delivery scheme proposed in Section \ref{CentralizedIntegerValues} corresponding to the normalized global cache capacity $\left\lfloor t \right\rfloor + 1$.  

For an arbitrary demand vector $\textbf{d} = \left( d_1, \dots, d_K\right)$, we define
\begin{equation}\label{DefCodedDeliveredContentCentralized12}
Q^i_{\mathcal{C}_l^{{t_i}+1}} \buildrel \Delta \over = {{\bigoplus}_{k \in \mathcal{C}_l^{{t_i}+1}} {W^i_{{d_k},\mathcal{C}_l^{{t_i}+1}\backslash \{ k\} }}}, \quad \mbox{for $i=1, 2$},  
\end{equation}
where $t_1 \buildrel \Delta \over = \left\lfloor t \right\rfloor$ and $t_2 \buildrel \Delta \over = \left\lfloor t \right\rfloor + 1$. According to Algorithm \ref{CentralizedSchemeAlg}, by performing the centralized caching and coded delivery scheme proposed in Section \ref{CentralizedIntegerValues} for the normalized global cache capacity $\left\lfloor t \right\rfloor$ to serve the subfiles $\left\{ W_1^1, \dots, W_N^1 \right\}$, each of rate $R_1$, user $k \in [K]$ should receive
\begin{equation}\label{CodedContentsUserkNotinLeadersCentralized1}
{\tilde Q}^1_k = \bigcup\nolimits_{l: {\mathcal{C}_l^{\left\lfloor t \right\rfloor +1}} \cap {\cal U}_{{\textbf{d}}} \ne \emptyset, {\mathcal{C}_l^{\left\lfloor t \right\rfloor +1}} \cap [k-1] = \emptyset, k \in {\mathcal{C}_l^{\left\lfloor t \right\rfloor +1}}} Q^1_{\mathcal{C}_l^{\left\lfloor t \right\rfloor +1}}, 
\end{equation}
of the following total rate obtained according to \eqref{TotalRateEachUserCentralized}
\begin{align}\label{TotalRateEachUserCentralizedNonInteger1}
{R_{{\textbf{d}},\left\lfloor t \right\rfloor,k}} = 
\begin{cases} 
\frac{\binom{K-k}{\left\lfloor t \right\rfloor}}{\binom{K}{\left\lfloor t \right\rfloor}}R_1, &\mbox{if $k \in \mathcal{U}_{{\textbf{d}}}$},\\
\frac{\binom{K-k}{\left\lfloor t \right\rfloor}-\binom{K-k-N_{{\textbf{d}},k}}{\left\lfloor t \right\rfloor}}{\binom{K}{\left\lfloor t \right\rfloor}}R_1, &\mbox{otherwise}.
\end{cases}
\end{align}
Similarly, to serve the subfiles $\left\{ W_1^2, \dots, W_N^2 \right\}$, each of rate $R_2$, user $k \in [K]$ should receive 
\begin{equation}\label{CodedContentsUserkNotinLeadersCentralized2}
{\tilde Q}^2_k = \bigcup\nolimits_{l: {\mathcal{C}_l^{\left\lfloor t \right\rfloor + 2}} \cap {\cal U}_{{\textbf{d}}} \ne \emptyset, {\mathcal{C}_l^{\left\lfloor t \right\rfloor + 2}} \cap [k-1] = \emptyset, k \in {\mathcal{C}_l^{\left\lfloor t \right\rfloor + 2}}} Q^2_{\mathcal{C}_l^{\left\lfloor t \right\rfloor + 2}}, 
\end{equation}
of total rate 
\begin{align}\label{TotalRateEachUserCentralizedNonInteger2}
{R_{{\textbf{d}},\left\lfloor t \right\rfloor + 1 ,k}} = 
\begin{cases} 
\frac{\binom{K-k}{\left\lfloor t \right\rfloor + 1}}{\binom{K}{\left\lfloor t \right\rfloor + 1}}R_2, &\mbox{if $k \in \mathcal{U}_{{\textbf{d}}}$},\\
\frac{\binom{K-k}{\left\lfloor t \right\rfloor + 1}-\binom{K-k-N_{{\textbf{d}},k}}{\left\lfloor t \right\rfloor + 1}}{\binom{K}{\left\lfloor t \right\rfloor + 1}}R_2, &\mbox{otherwise}.
\end{cases}
\end{align}
Thus, ${\tilde Q}_k \buildrel \Delta \over = \left( {\tilde Q}^1_k,{\tilde Q}^2_k \right)$, at a total rate of ${R^{\rm{C}}_{{\textbf{d}},k}} \buildrel \Delta \over = {R_{{\textbf{d}},\left\lfloor t \right\rfloor,k}} + {R_{{\textbf{d}},\left\lfloor t \right\rfloor + 1 ,k}}$ given by
\begin{align}\label{TotalRateUserkCentralized}
{R^{\rm{C}}_{{\textbf{d}},k}} =
\begin{cases} 
\frac{\binom{K-k}{\left\lfloor t \right\rfloor}}{\binom{K}{\left\lfloor t \right\rfloor}}\left( \left\lfloor t \right\rfloor +1 -t \right)R+\\
\qquad \frac{\binom{K-k}{\left\lfloor t \right\rfloor+1}}{\binom{K}{\left\lfloor t \right\rfloor+1}}\left(t- \left\lfloor t \right\rfloor \right)R, &\mbox{if $k \in \mathcal{U}_{{\textbf{d}}}$},\\
\frac{\binom{K-k}{\left\lfloor t \right\rfloor}-\binom{K-k-N_{{\textbf{d}},k}}{\left\lfloor t \right\rfloor}}{\binom{K}{\left\lfloor t \right\rfloor}}\left( \left\lfloor t \right\rfloor +1 -t \right)R +\\
\qquad \frac{\binom{K-k}{\left\lfloor t \right\rfloor+1}-\binom{K-k-N_{{\textbf{d}},k}}{\left\lfloor t \right\rfloor+1}}{\binom{K}{\left\lfloor t \right\rfloor+1}}\left(t- \left\lfloor t \right\rfloor \right)R, &\mbox{otherwise}
\end{cases}
\end{align}
is delivered to user $k$, for $k=1, ..., K$.

\subsection{Transmit Power Analysis}\label{PowerAnalysisCentralized}

%For a demand vector $\textbf{d}$ in the delivery phase, the proposed scheme delivers a message of rate ${R^{\rm{C}}_{\textbf{d},k}}$, given in \eqref{TotalRateUserkCentralized}, to user $k \in [K]$, equipped with a cache of capacity $M \in [0,N]$ through linear superposition coding $\sum\nolimits_{k = 1}^K {x_{k}^n \left(\textbf{W},\textbf{d}, {\tilde q}_k \right)}$, and the received signal at user $k \in [K]$ is
%\begin{equation}\label{ReceivedUserkCent} y^n_{k} \left( \textbf{W},\textbf{d} \right) = h_k \sum\limits_{k = 1}^K {x_{k}^n \left( \textbf{W},\textbf{d}, {\tilde q}_k \right)} + {z_{k}^n}, \; \mbox{where ${\tilde q}_k \in \left[ 2^{n{R^{\rm{C}}_{\textbf{d},k}}} \right]$}. 
%\end{equation}

In the proposed delivery scheme, user $k \in [K]$ decodes codewords $x_{1}^n \left( \textbf{W},\textbf{d} \right), \dots, x_{k}^n \left( \textbf{W},\textbf{d} \right)$ successively considering all the other codewords in higher levels as noise. User $k$ can decode its intended message successfully if, for $k=1, ..., K$, 
\begin{equation}\label{RateRegionCentNonInt}
{R^{\rm{C}}_{\textbf{d},k}} \le \frac{1}{2}{\log _2}\left( {1 + \frac{{{\alpha _k} h_k^2 P^{\rm{C}}_{{\rm{UB}}}\left( R,M,\textbf{d} \right)}}{{h_k^2 \sum\nolimits_{i = k + 1}^K {{\alpha _i}P^{\rm{C}}_{{\rm{UB}}}\left( R,M,\textbf{d} \right)}  + 1}}} \right).
\end{equation}
From Proposition \ref{GeneralPowerAllocationLemma}, the corresponding minimum required power is given by
\begin{equation}\label{RequiredPowerEachDemandDeliveryCentralized}
P^{\rm{C}}_{{\rm{UB}}}\left( R,M,\textbf{d} \right) \buildrel \Delta \over = \sum\nolimits_{i = 1}^K \left( \frac{2^{2{R^{\rm{C}}_{{\textbf{d}},i}}} - 1}{h_i^2} \right)\prod\nolimits_{j = 1}^{i - 1} {{2^{2{R^{\rm{C}}_{{\textbf{d}},j}}}}} .
\end{equation}
Thus, the average power-memory trade-off of the proposed achievable scheme is given by ${{\rm E}_{\textbf{d}}}\left[ P^{\rm{C}}_{\rm{UB}}\left( R,M,\textbf{d} \right) \right] = \bar P^{\rm{C}}_{\rm{UB}}(R,M)$, while the peak power-memory trade-off is $\hat P^{\rm{C}}_{\rm{UB}} \left( R,M \right) = \mathop {\max }\limits_{\textbf{d}} \left\{ {P^{\rm{C}}_{\rm{UB}}\left( {R,M,{\textbf{d}}} \right)} \right\}$, as stated in Theorem \ref{UpperboundPowerMemoryTheoremCentralized}, where the demands are distributed uniformly.

Observe that, for demand vectors with the same ${\cal U}_{{\textbf{d}}}$ set, the required power $P^{\rm{C}}_{\rm{UB}}\left( R,M,\textbf{d} \right)$ is the same. Let $\mathcal{D}_{\mathcal{U}_{\textbf{d}}}$ denote set of all demand vectors with the same $\mathcal{U}_{\textbf{d}}$ set. We define $P^{\rm{C}}_{\rm{UB}}\left( R,M,\mathcal{D}_{\mathcal{U}_{\textbf{d}}} \right)$ as the required power $P^{\rm{C}}_{\rm{UB}}\left( R,M,\textbf{d} \right)$ for any demand vector $\textbf{d} \in \mathcal{D}_{\mathcal{U}_{\textbf{d}}}$. Thus, we have
\begin{equation}\label{RequiredPowerEachDemandDeliverySameSetCentralized}
\hat P^{\rm{C}}_{\rm{UB}} \left( R,M \right) = \mathop {\max }\limits_{{\cal U}_\textbf{d}} \left\{ {P^{\rm{C}}_{\rm{UB}}\left( {R,M,\mathcal{D}_{\mathcal{U}_{\textbf{d}}}} \right)} \right\}.
\end{equation}
It is shown in Appendix \ref{ProofWorstCaseDemandCentralized} that the worst-case demand combination for the proposed centralized caching scheme happens when the first $\min\{N,K\}$ users; that is, the users with the worst channel gains, request distinct files, i.e., when $\mathcal{U}_{\textbf{d}} = \left[ \min\{N,K\} \right]$, and $\hat P^{\rm{C}}_{\rm{UB}} \left( R,M \right)$ is given by \eqref{AchievablePowerMemoryTheoremWorstCaseCentralized}.
\end{proof}

\section{Decentralized Caching and Delivery}\label{ProposedSchemeDecentralized}

Here we extend our centralized caching scheme to the decentralized caching. The corresponding average and peak power-memory trade-offs are given in the following theorem.

\begin{theorem}\label{UpperboundPowerMemoryTheoremDecentralized}
For decentralized caching followed by delivery over a Gaussian BC, we have
\begin{subequations}
\label{AchievablePowerMemoryTheoremDemandDecentralized}
\begin{align}\label{AchievablePowerMemoryTheoremDemandDecentralized1}
& \bar P^*(R,M) \le \bar P^{\rm{D}}_{\rm{UB}}(R,M) \buildrel \Delta \over =\nonumber\\
&\qquad \quad \frac{1}{N^K} \sum\limits_{\emph{\textbf{d}} \in \left[ N \right]^K}  \left[ \sum\limits_{i = 1}^K \left( \frac{{{2^{2{R^{\rm{D}}_{\emph{\textbf{d}},i}}}} - 1}}{h_i^2} \right)\prod\limits_{j = 1}^{i - 1} {{2^{2{R^{\rm{D}}_{\emph{\textbf{d}},j}}}}}  \right] ,
\end{align}
where, for $k=1, ..., K$, 
\begin{align}\label{AchievablePowerMemoryTheoremDemandDecentralized2}
{R^{\rm{D}}_{\emph{\textbf{d}},k}} \buildrel \Delta \over = 
\begin{cases} 
{\left( {1 - \frac{M}{N}} \right)^k}R, &\mbox{if $k \in \mathcal{U}_{\emph{\textbf{d}}}$},\\
{\left( {1 - \frac{M}{N}} \right)^k}\left( {1 - {{\left( {1 - \frac{M}{N}} \right)}^{{N_{\emph{\textbf{d}},k}}}}} \right)R, &\mbox{otherwise},
\end{cases}
\end{align}
\end{subequations}
and 
\begin{align}\label{AchievablePowerMemoryTheoremWorstCaseTheoremDecentralized}
\hat P^*(R,M) \le \hat P^{\rm{D}}_{\rm{UB}}(R,M) \buildrel \Delta \over = & \sum\limits_{i = 1}^{\min\{N,K\}} \left( \frac{{{2^{2R{\left( {1 - \frac{M}{N}} \right)^i}}} - 1}}{h_i^2} \right) \nonumber\\
& 2^{2R\left( {\frac{N}{M} - 1} \right)\left( {1 - {{\left( {1 - \frac{M}{N}} \right)}^{i - 1}}} \right)} .
\end{align}
\end{theorem}

\begin{proof}
The decentralized caching scheme to achieve the average and peak power-memory trade-offs outlined in Theorem \ref{UpperboundPowerMemoryTheoremDecentralized} is described in the following.

\subsection{Placement Phase}\label{ProposedSchemePlacement}
We perform decentralized uncoded cache placement \cite{MaddahAliDecentralized}, where each user caches $nMR/N$ random bits of each file of length $nR$ bits independently. Since there are a total of $N$ files, the cache capacity constraint is satisfied. The part of file $i$ cached exclusively by the users in set $\mathcal S \subset \left[ K \right]$ is denoted by $W_{i,\mathcal S}$, for $i=1,..., N$. For $n$ large enough, the rate of $W_{i,\mathcal S}$ can be approximated by ${\left( {\frac{M}{N}} \right)^{\left| \mathcal S \right|}}{\left( {1 - \frac{M}{N}} \right)^{K - \left| \mathcal S \right|}}R$. The cache contents at user $k$ is given by
\begin{equation}\label{CacheContentUserk}
{U_k} = \bigcup\nolimits_{i \in \left[ N \right]} {\bigcup\nolimits_{\mathcal{S} \subset \left[K\right]: k \in {\cal S}} {W_{i,{\cal S}}} }.
\end{equation}

\subsection{Delivery phase}\label{ProposedSchemeDelivery}
Consider any non-empty set of users $\mathcal S \subset \left[ K \right]$. For a demand vector ${\textbf{d}}$, by delivering the coded message $Q_{{\cal S}} = {{\bigoplus}_{k \in \mathcal{S}} {W_{{d_k},\mathcal{S}\backslash \{ k\} }}} $ of rate ${\left( {\frac{M}{N}} \right)^{\left| \mathcal S \right|-1}}{\left( {1 - \frac{M}{N}} \right)^{K - \left| \mathcal S \right|+1}}R $ to users in $\cal S$, each user $i \in {\cal S}$ can recover subfile $W_{{d_i},\mathcal{S}\backslash \{ i\} }$ since it has cached all the subfiles $W_{{d_j},\mathcal{S}\backslash \{ j\} }$, $\forall j \in {\cal S} \backslash \{i\}$. For each $k \in [K]$, delivering $\bigcup\nolimits_{\mathcal S \subset \left[ K \right]:k \in \mathcal {S}} Q_{\mathcal{S}}$ enables user $k$ to recover all the subfiles $W_{{d_k},\mathcal{S}\backslash \{k\}}$, $\forall \mathcal{S} \subset \left[ K \right]$ and $k \in \cal S$. The demand of user $k$, $k \in [K]$, is fully satisfied after receiving $\bigcup\nolimits_{\mathcal S \subset \left[ K \right]:k \in \mathcal {S}} Q_{\mathcal{S}}$ along with its cache contents.

Similarly to the proposed scheme for the centralized caching scenario, given a demand vector $\textbf{d}$, the delivery phase is designed such that only the coded messages $Q_{\cal S}$, $\forall \mathcal S \subset [K]$ that satisfy $\mathcal S \cap {\cal U}_{{\textbf{d}}} \ne \emptyset$, are delivered, and the remaining coded messages can be recovered through \eqref{LemmaEqConcLast}. To achieve this, for any such set $\cal S$ with $\mathcal S \cap {\cal U}_{{\textbf{d}}} \ne \emptyset$, the transmission power is adjusted such that the worst user in ${\cal S}$ can decode it; and so can all the other users in $\cal S$ due to the degradedness of the Gaussian BC. Therefore, the demand of every user in $\cal{U}_{\textbf{d}}$ is satisfied.

Note that in the centralized scenario described in Section \ref{CentralizedIntegerValues}, each coded packet $Q_{\mathcal{C}_l^{{t}+1}}$ is targeted for a $(t+1)$-element subset of users, where $t \in [0:K]$, for $l \in \left[ \binom{K}{t+1} \right]$. While, in the decentralized scenario, coded packets are targeted for any subset of users, i.e., for $(t+1)$-element subset of users, $\forall t \in [0:K]$. By applying a similar technique as the delivery phase outlined in Algorithm \ref{CentralizedSchemeAlg}, given a demand vector $\textbf{d}$, the following contents are targeted for user $k$, $k \in [K]$:
\begin{equation}\label{CodedContentsUserkNotinLeaders}
{\tilde Q}_k = \bigcup\nolimits_{\mathcal S \subset [k:K]:\mathcal S \cap {\cal U}_{{\textbf{d}}} \ne \emptyset,k \in {\cal S}} {{Q_{{\cal S}}}}, 
\end{equation}
which are equivalent to
\begin{equation}\label{CodedContentsUserkNotinLeadersEquivalent}
{\tilde Q}_k = \bigcup\nolimits_{\mathcal S \subset [k:K]:k \in {\cal S}} {{Q_{{\cal S}}}} - \bigcup\nolimits_{\mathcal S \subset [k:K]\backslash {\mathcal{U}_{\textbf{d},k}},k \in {\cal S}} {{Q_{{\cal S}}}}. 
\end{equation}
For each user $k$, $k \in [K]$, there are $\binom{K-k}{i}$ different $(i+1)$-element subsets of $[k:K]$, which include $k$, for $i \in [0:K-k]$. Thus, if $k \in \cal{U}_{\textbf{d}}$, from \eqref{CodedContentsUserkNotinLeadersEquivalent}, the total rate targeted for user $k$ is
\begin{align}\label{TotalRateIntendedUserkinLeaders}
R^{\rm{D}}_{\textbf{d},k} &= \sum\limits_{i = 0}^{K-k} {\binom{K-k}{i}{\left( {\frac{M}{N}} \right)}^{i }}{{\left( {1 - \frac{M}{N}} \right)}^{K - i }}R \nonumber\\
& = \left(1 - \frac{M}{N}\right)^kR.
\end{align}
On the other hand, for each user $k$, $k \in [K]$, there are $\binom{K-k-N_{{\textbf{d}},k}}{i}$ different $(i+1)$-element subsets of $[k:K] \backslash {\mathcal{U}_{\textbf{d},k}}$, which include $k$, for $i \in [0:K-k]$. Thus, if $k \notin \cal{U}_{\textbf{d}}$, from \eqref{CodedContentsUserkNotinLeadersEquivalent}, the total rate targeted for user $k \in [K]$ is given by
\begin{align}\label{TotalRateIntendedUserkNotinLeaders}
R^{\rm{D}}_{\textbf{d},k} =& \sum\limits_{i = 0}^{K-k} {\binom{K-k}{i}{\left( {\frac{M}{N}} \right)}^{i }}{{\left( {1 - \frac{M}{N}} \right)}^{K - i }}R\nonumber\\
& - \sum\limits_{i = 0}^{K-k-{N_{{\textbf{d}},k}}} {\binom{K-k-{N_{{\textbf{d}},k}}}{i}{\left( {\frac{M}{N}} \right)}^{i }}{{\left( {1 - \frac{M}{N}} \right)}^{K - i }}R\nonumber\\
= & {\left( {1 - \frac{M}{N}} \right)^k}\left( {1 - {{\left( {1 - \frac{M}{N}} \right)}^{{N_{{\textbf{d}},k}}}}} \right)R.
\end{align}
In total, the rate of contents targeted for user $k$, $k \in [K]$, is given by
\begin{align}\label{TotalRateEachUser}
{R^{\rm{D}}_{\textbf{d},k}} = 
\begin{cases} 
{\left( {1 - \frac{M}{N}} \right)^k}R, & \mbox{if $k \in \mathcal{U}_{{\textbf{d}}}$},\\
{\left( {1 - \frac{M}{N}} \right)^k}\left( {1 - {{\left( {1 - \frac{M}{N}} \right)}^{{N_{{\textbf{d}},k}}}}} \right)R, &\mbox{otherwise}.
\end{cases}
\end{align}

Given a demand vector $\textbf{d}$, the transmitted codeword $x^n \left( \textbf{W},\textbf{d} \right)$ is generated as the linear superposition of $K$ codewords $x_{1}^n \left( \textbf{W},\textbf{d} \right), ..., x_{K}^n \left( \textbf{W},\textbf{d} \right)$, each chosen from an independent codebook. Codebook $k$ consists of $2^{n{R^{\rm{D}}_{{\textbf{d}},k}}}$ i.i.d. codewords $x_{k}^n \left( \textbf{W},\textbf{d} \right)$ generated according to the normal distribution $\mathcal{N} \left( 0,\alpha_k P^{\rm{D}}_{{\rm{UB}}}\left( R,M,\textbf{d} \right) \right)$, where $\alpha_k \ge 0$ and $\sum\nolimits_{i = 1}^K {{\alpha _i}}  = 1$, which satisfy the power constraint, for $k=1, ..., K$.

User $k$, $k \in [K]$, decodes codewords $x_{1}^n \left( \textbf{W},\textbf{d} \right)$, $...$, $x_{k}^n \left( \textbf{W},\textbf{d} \right)$ through successive decoding, while considering all the codewords $x_{k+1}^n \left( \textbf{W},\textbf{d} \right), ..., x_{K}^n \left( \textbf{W},\textbf{d} \right)$ as noise. Thus, if user $k \in [K]$ can successfully decode all the $k$ channel codewords intended for it, it can then recover the contents ${\tilde Q}_1, ..., {\tilde Q}_k$, from which it can obtain all the coded contents $\bigcup\nolimits_{{{\mathcal S}} \cap {\cal U}_{{\textbf{d}}} \ne \emptyset, [k] \cap {{\mathcal S}} \ne \emptyset} Q_{{\mathcal S}}$. Accordingly, each user $k \in [K]$ can obtain all the coded contents targeted for it except those that are not intended for at least one user in ${\cal U}_{{\textbf{d}}}$ (which have not been delivered), i.e., all the coded contents
\begin{equation}\label{CodedContentsDecNewUserkDecentk}
\bigcup\nolimits_{{{\mathcal S}} \cap {\cal U}_{{\textbf{d}}} \ne \emptyset, k \in {{\mathcal S}}} Q_{{\mathcal S}}.
\end{equation}
It can further obtain all the coded contents targeted for users $[k-1]$, which are also intended for at least one user in ${\cal U}_{{\textbf{d}}}$, i.e., all the coded contents $\bigcup\nolimits_{{\mathcal{S}} \cap {\cal U}_{{\textbf{d}}} \ne \emptyset, [k-1] \cap {{\mathcal{S}}}\ne \emptyset} Q_{{\mathcal{S}}}$. Note that, if $k \in {\cal U}_{{\textbf{d}}}$, \eqref{CodedContentsDecNewUserkDecentk} reduces to $\bigcup\nolimits_{k \in {{\mathcal{S}}}} Q_{{\mathcal{S}}}$, which shows that the demand of each user $k \in {\cal U}_{{\textbf{d}}}$ is satisfied through the proposed delivery scheme. Next, we illustrate that the users in $[K] \backslash {\cal U}_{{\textbf{d}}}$ can obtain their requested files without being delivered any extra messages. Similarly to the centralized scenario, given any set of users ${\mathcal{S}}$ such that ${\mathcal{S}} \cap {\cal U}_{{\textbf{d}}} = \emptyset$, by setting $\mathcal B ={\mathcal{S}} \cup {\cal U}_{{\textbf{d}}}$, from the fact that, for each user $k \in {\mathcal{S}}$, $k \in \mathcal B \backslash \cal G$ or $k' \in \mathcal B \backslash \cal G$, where $k'<k$, it can be illustrated that every user in ${\mathcal{S}}$ can decode all coded contents $Q_{\mathcal B \backslash \cal G}$, $\forall \mathcal G \in {\mathcal {G}_{\cal B}}\backslash {\mathcal U_{{\textbf{d}}}}$. In this case, they all can also decode $Q_{\mathcal{S}}$ through \eqref{LemmaEqConcLast}.

\subsection{Transmit Power Analysis}\label{PowerAnalysisDecentralized}
For a demand vector $\textbf{d}$, user $k$ can decode the channel codewords up to level $k$ successfully, considering all the other codewords in higher levels as noise, if, for $k=1, ..., K$, 
\begin{equation}\label{RateRegionDecentNew}
{R^{\rm{D}}_{\textbf{d},k}} \le \frac{1}{2}{\log _2}\left( {1 + \frac{{{\alpha _k} h_k^2 P^{\rm{D}}_{{\rm{UB}}}\left( R,M,\textbf{d} \right)}}{{h_k^2 \sum\nolimits_{i = k + 1}^K {{\alpha _i}P^{\rm{D}}_{{\rm{UB}}}\left( R,M,\textbf{d} \right)}  + 1}}} \right).
\end{equation}
From Proposition \ref{GeneralPowerAllocationLemma}, the minimum required power is given by
\begin{equation}\label{RequiredPowerEachDemandDelivery}
P^{\rm{D}}_{\rm{UB}}\left( R,M,\textbf{d} \right) \buildrel \Delta \over = \sum\nolimits_{i = 1}^K {\left( \frac{{{2^{2{R^{\rm{D}}_{{\textbf{d}},i}}}} - 1}}{h_i^2} \right)\prod\nolimits_{j = 1}^{i - 1} {{2^{2{R^{\rm{D}}_{{\textbf{d}},j}}}}} }.
\end{equation}
Thus, the average power-memory trade-off for the proposed decentralized caching and coded delivery scheme is given by ${{\rm E}_{\textbf{d}}}\left[ P^{\rm{D}}_{\rm{UB}}\left( R,M,\textbf{d} \right) \right] = \bar P^{\rm{D}}_{\rm{UB}}(R,M)$ stated in Theorem \ref{UpperboundPowerMemoryTheoremDecentralized}.

With the proposed decentralized caching scheme, the peak power performance $\hat P^{\rm{D}} \left( R,M \right) = \mathop {\max }\limits_{\textbf{d}} \left\{ {P^{\rm{D}}_{\rm{UB}}\left( {R,M,{\textbf{d}}} \right)} \right\}$ can be achieved. Observe that, for demand vectors with the same set of users ${\cal U}_{{\textbf{d}}}$, the required power $P^{\rm{D}}_{\rm{UB}}\left( R,M,\textbf{d} \right)$ is the same. We define $P^{\rm{D}}_{\rm{UB}}\left( R,M,\mathcal{D}_{\mathcal{U}_{\textbf{d}}} \right)$ as the required power $P^{\rm{D}}_{\rm{UB}}\left( R,M,\textbf{d} \right)$ for any demand vector $\textbf{d} \in \mathcal{D}_{\mathcal{U}_{\textbf{d}}}$. Thus, we have
\begin{equation}\label{RequiredPowerEachDemandDeliverySameSet}
\hat P^{\rm{D}} \left( R,M \right) = \mathop {\max }\limits_{{\cal U}_\textbf{d}} \left\{ {P^{\rm{D}}_{\rm{UB}}\left( {R,M,\mathcal{D}_{\mathcal{U}_{\textbf{d}}}} \right)} \right\}.
\end{equation}
It is shown in Appendix \ref{ProofWorstCaseDemand} that the worst-case demand combination happens when $\mathcal{U}_{\textbf{d}} = \left[ \min\{N,K\} \right]$, and $\hat P^{\rm{D}}_{\rm{UB}} \left( R,M \right)$ is found as in \eqref{AchievablePowerMemoryTheoremWorstCaseTheoremDecentralized}. 
\end{proof}

\section{Lower Bound}\label{ProofSecondTheorem}

If we constrain the placement phase to uncoded caching, we can lower bound $\bar P^*\left( {R,M} \right)$ and ${{\hat P}^*}\left( {R,M} \right)$ as in the following theorem. The main challenge in deriving a lower bound for the cache-aided BC studied here is the lack of degradedness due to the presence of the caches. To derive a lower bound, we assume that the files requested by users in $[k-1]$ and their cache contents are provided to the other users. We then exploit the degradedness of the resultant system to lower bound the performance of the original model.

\begin{theorem}\label{LowerboundPowerMemoryTheorem}
In cache-aided content delivery over a Gaussian BC with uncoded cache placement phase, the minimum average power is lower bounded by $\bar P_{\rm{LB}}(R,M)$ defined as 
\begin{align}\label{LowerBoundAveragePowerMemoryTheorem}
\bar P_{\rm{LB}}(R,M) \buildrel \Delta \over = & {{\rm{E}}_{{\mathcal{U}_{\emph{\textbf{d}}}}}} \left[ \sum\nolimits_{i = 1}^{N_{\emph{\textbf{d}}}} \left( \frac{{{2^{2R\left( {1 - \min \left\{ {iM/N,1} \right\}} \right)}} - 1}}{h_{\pi_{\mathcal{U}_{\emph{\textbf{d}}}}(i)}^2} \right) \right.\nonumber\\
& \qquad \qquad \left. \prod\nolimits_{j = 1}^{i - 1} 2^{2R\left( {1 - \min \left\{ {jM/N,1} \right\}} \right)}  \right],
\end{align}
where ${{\rm{E}}_{{\mathcal{U}_{\emph{\textbf{d}}}}}} [\cdot]$ takes the expectation over all possible sets ${\mathcal{U}_{\emph{\textbf{d}}}}$, and $\pi_{\cal S}$ is a permutation over any subset of users ${\cal S} \subset [K]$, such that $h_{\pi_{\mathcal S}(1)}^2 \le h_{\pi_{\mathcal S}(2)}^2 \le \cdots \le h_{\pi_{\mathcal S}(\left| \mathcal S \right|)}^2$. The minimal required peak transmit power for the same system is lower bounded by $\hat P_{\rm{LB}}(R,M)$ defined as
\begin{align}\label{LowerBoundWorstCasePowerMemoryTheorem}
\hat P_{\rm{LB}}(R,M) \buildrel \Delta \over = & \mathop {\max }\limits_{\mathcal S \subset \left[ {\min \left\{ {N,K} \right\}} \right]} \left\{ \sum\limits_{i = 1}^{\left| {\cal S} \right|} \left( \frac{{{2^{2R\left( {1 - \min \left\{ {iM/N,1} \right\}} \right)}} - 1}}{h_{{\pi _S}(i)}^2} \right) \right. \nonumber\\
& \qquad \qquad \left. \prod\nolimits_{j = 1}^{i - 1} 2^{2R\left( {1 - \min \left\{ {jM/N,1} \right\}} \right)}  \right\}.
\end{align}
\end{theorem}

\begin{remark}\label{SimplificationAverageLowerBound}
Considering all possible demand vectors, we have a total of $2^{K-1}$ different $\mathcal{U}_{\emph{\textbf{d}}}$ sets. This follows from the fact that $N_{\emph{\textbf{d}}} \le K$ and $1 \in \mathcal{U}_{\emph{\textbf{d}}}$, $\forall \emph{\textbf{d}}$. For a given demand vector $\emph{\textbf{d}}$, let $\mathcal{U}_{\emph{\textbf{d}}} = \left\{ u_1, u_2, ..., u_{N_{\emph{\textbf{d}}}} \right\}$, where $1=u_1 \le u_2 \le \cdots \le u_{N_{\emph{\textbf{d}}}}$. The number of demand vectors with the same $\mathcal{U}_{\emph{\textbf{d}}}$ is given by
\begin{equation}\label{NumberDemandsInUd}
N_{\mathcal{U}_{\emph{\textbf{d}}}} \buildrel \Delta \over = \binom{N}{N_{\emph{\textbf{d}}}}{N_{\emph{\textbf{d}}}!} \left( \prod\nolimits_{j = 2}^{N_{\emph{\textbf{d}}}} j^{u_{j+1} - u_{j} -1} \right),
\end{equation}
where we define $u_{N_{\emph{\textbf{d}}}+1} \buildrel \Delta \over = K+1$. Thus, the lower bound in \eqref{LowerBoundAveragePowerMemoryTheorem} reduces to
\begin{align}\label{SimplificationAverageLowerBoundEq}
& \bar P_{\rm{LB}}(R,M) = \frac{1}{N^K} \sum\nolimits_{\mathcal U_{\emph{\textbf{d}}} \subset \left[ K \right], 1 \in \mathcal U_{\emph{\textbf{d}}}} N_{\mathcal{U}_{\emph{\textbf{d}}}}  \nonumber\\
& \sum\limits_{i = 1}^{N_{\emph{\textbf{d}}}}  \left( \frac{{{2^{2R\left( {1 - \min \left\{ {iM/N,1} \right\}} \right)}} - 1}}{h_{\pi_{\mathcal{U}_{\emph{\textbf{d}}}}(i)}^2} \right)\prod\limits_{j = 1}^{i - 1} 2^{2R\left( {1 - \min \left\{ {jM/N,1} \right\}} \right)}.
\end{align}
\end{remark}

\begin{proof}
For any given $\mathcal{U}_{\textbf{d}}$, let $\mathcal{D}_{\mathcal{U}_{\textbf{d}}}$ denote the set of all demand vectors with the same $\mathcal{U}_{\textbf{d}}$. The union $\bigcup\nolimits_{\mathcal{U}_{\textbf{d}}} \mathcal{D}_{\mathcal{U}_{\textbf{d}}}$ form the set of all possible demand vectors $[N]^K$. Therefore, the set of all possible demand vectors can be broken into classes $\mathcal{D}_{\mathcal{U}_{\textbf{d}}}$ based on the $\mathcal{U}_{\textbf{d}}$ set they correspond to. 

For any given $\mathcal{U}_{\textbf{d}}$, and an $(n,R,M)$ code as defined in \eqref{CachingFunction}, \eqref{DeliveryFunction} and \eqref{DecodingFunctionWeak} in Section \ref{SystemModel}, define the error probability as follows:
\begin{equation}\label{ErrorProbabilityFixedUd} 
{P_{{e}_{\mathcal{U}_{\textbf{d}}}}} \buildrel \Delta \over = \Pr \left\{ \bigcup\nolimits_{\textbf{d} \in \mathcal{D}_{\mathcal{U}_{\textbf{d}}}}{{\bigcup\nolimits_{k \in \mathcal{U}_{\textbf{d}}} {\left\{ {{{\hat W}_{{d_k}}} \ne {W_{{d_k}}}} \right\}} }} \right\}.
\end{equation}
Let ${P_{\mathcal{U}_{\textbf{d}}}\left( \textbf{d} \right)}$ denote the average power of the codeword this code generates for a demand vector $\textbf{d} \in \mathcal{D}_{\mathcal{U}_{\textbf{d}}}$. We say that an $\left( R,M,\bar P, \hat P \right)$ tuple is $\mathcal{U}_{\textbf{d}}$-\textit{achievable} if for every $\varepsilon > 0$, there exists an $(n,R,M)$ code with sufficiently large $n$, which satisfies ${P_{{e}_{\mathcal{U}_{\textbf{d}}}}} < \varepsilon$, ${{\rm E}_{\textbf{d}}}\left[ {P_{\mathcal{U}_{\textbf{d}}}\left( \textbf{d} \right)} \right] \le \bar P$, and ${P_{\mathcal{U}_{\textbf{d}}}\left( \textbf{d} \right)} \le \hat P$, $\forall \textbf{d} \in \mathcal{D}_{\mathcal{U}_{\textbf{d}}}$. We can also define ${\bar P}^*_{\mathcal{U}_{\textbf{d}}} \left( {R,M} \right)$ and ${\hat P}^*_{\mathcal{U}_{\textbf{d}}}\left( {R,M} \right)$ as in \eqref{AveragePowerMemoryTradeOff} and \eqref{PeakPowerMemoryTradeOff}, respectively, by considering $\mathcal{U}_{\textbf{d}}$-achievable codes.

We note from \eqref{ErrorProbabilityFixedUd} that, a $\mathcal{U}_{\textbf{d}}$-achievable code satisfies only the demands of the users in set $\mathcal{U}_{\textbf{d}}$. Accordingly, an achievable $\left( R,M,\bar P, \hat P \right)$ tuple is also $\mathcal{U}_{\textbf{d}}$-achievable, since $P_e \ge {P_{{e}_{\mathcal{U}_{\textbf{d}}}}}$, for any ${\mathcal{U}_{\textbf{d}}}$ set. Thus, lower bounds on ${\bar P}^*_{\mathcal{U}_{\textbf{d}}} \left( {R,M} \right)$ and ${\hat P}^*_{\mathcal{U}_{\textbf{d}}}\left( {R,M} \right)$ also serve as lower bounds on ${\bar P}^* \left( {R,M} \right)$ and ${\hat P}^* \left( {R,M} \right)$, respectively. In the following, we provide lower bounds on ${\bar P}^*_{\mathcal{U}_{\textbf{d}}} \left( {R,M} \right)$ and ${\hat P}^*_{\mathcal{U}_{\textbf{d}}}\left( {R,M} \right)$.

Let $\left( R,M,\bar P, \hat P \right)$ be any $\mathcal{U}_{\textbf{d}}$-achievable tuple. For uniformly distributed demands, we have
\begin{equation}\label{SmallerSubsetsPowerMemory}
\bar P \ge {{\rm E}_{\textbf{d}}}\left[ {P_{\mathcal{U}_{\textbf{d}}}\left( \textbf{d} \right)} \right] = {{\rm{E}}_{{\mathcal{U}_{\textbf{d}}}}}\left[ \frac{1}{N_{\mathcal{U}_{{\textbf{d}}}}} \sum\limits_{{\textbf{d}} \in {\mathcal{D}_{{\mathcal{U}_{\textbf{d}}}}}} {P_{\mathcal{U}_{\textbf{d}}}\left( \textbf{d} \right)} \right],
\end{equation}
where we used the fact that the probability of each demand vector in $\mathcal{D}_{\mathcal{U}_{\textbf{d}}}$ is equal. We divide the set of demand vectors $\mathcal{D}_{\mathcal{U}_{\textbf{d}}}$ into different subsets according to the demands of users in ${\mathcal{U}_{\textbf{d}}}$, where each subset consists of the demand vectors for which the demands of all the users in $\mathcal{U}_{\textbf{d}}$ are the same. Note that, there are $D \buildrel \Delta \over = \binom{N}{N_{\textbf{d}}}{N_{\textbf{d}}!}$ such subsets\footnote{For simplicity, we drop the dependence of $D$ on $N$ and $N_{\textbf{d}}$.}, denoted by $\mathcal{D}^l_{\mathcal{U}_{\textbf{d}}}$, for $l=1, ..., D$, i.e., $\mathcal{D}_{\mathcal{U}_{\textbf{d}}} = {\bigcup\nolimits_{l = 1}^{D} \mathcal{D}^l_{\mathcal{U}_{\textbf{d}}} }$. We note that the number of demand vectors in each ${{{\cal D}^l_{{{\cal U}_{\textbf{d}}}}}}$, denoted by $N'_{\mathcal{U}_{{\textbf{d}}}}$, is the same, and is given by
\begin{equation}\label{NumberDemandVectorsl}
N'_{\mathcal{U}_{{\textbf{d}}}} = \prod\nolimits_{j = 2}^{N_{{\textbf{d}}}} j^{u_{j+1} - u_{j} -1},
\end{equation}
where, we remind that $\mathcal{U}_{{\textbf{d}}} = \left\{ u_1, u_2, ..., u_{N_{{\textbf{d}}}} \right\}$, where $1=u_1 \le u_2 \le \cdots \le u_{N_{{\textbf{d}}}}$. Thus, we have $N_{\mathcal{U}_{{\textbf{d}}}} = D N'_{\mathcal{U}_{{\textbf{d}}}}$, and \eqref{SmallerSubsetsPowerMemory} can be rewritten as follows:
\begin{align}\label{SmallerSubsetsPowerMemory2Rewritten}
\bar P &\ge  {{\rm{E}}_{{\mathcal{U}_{\textbf{d}}}}}\left[ \frac{1}{N_{\mathcal{U}_{{\textbf{d}}}}} \sum\nolimits_{{\textbf{d}} \in {\mathcal{D}_{{\mathcal{U}_{\textbf{d}}}}}} {P_{\mathcal{U}_{\textbf{d}}}\left( \textbf{d} \right)} \right] \nonumber\\
& = {{\rm{E}}_{{\mathcal{U}_{\textbf{d}}}}}\left[ \frac{1}{D} \sum\nolimits_{l=1}^{D} \left( \frac{1}{N'_{\mathcal{U}_{{\textbf{d}}}}} \sum\nolimits_{{\textbf{d}} \in {\mathcal{D}^l_{{\mathcal{U}_{\textbf{d}}}}}} {P_{\mathcal{U}_{\textbf{d}}}\left( \textbf{d} \right)} \right) \right].
\end{align}
For any arbitrary demand vector $\textbf{d}^{l}_{\mathcal{U}_{\textbf{d}}} \in \mathcal{D}^l_{\mathcal{U}_{\textbf{d}}}$, for $l \in [D]$, it is proved in \cite[Lemma 14]{ShirinWiggerYenerCacheAssingment} that there exist random variables\footnote{For ease of presentation, we drop the dependence of the transmitted signal $X^n$, and the received signals $Y_k^n$, $\forall k \in [K]$, on the library $\textbf{W}$.} $X \left( \textbf{d}^{l}_{\mathcal{U}_{\textbf{d}}} \right)$, $Y_{{{\pi_{{\cal U}_{\textbf{d}}(1)}}}} \left( \textbf{d}^{l}_{\mathcal{U}_{\textbf{d}}} \right), \dots, Y_{{{\pi_{{\cal U}_{\textbf{d}}(N_{\textbf{d}})}}}} \left( \textbf{d}^{l}_{\mathcal{U}_{\textbf{d}}} \right)$, and $\left\{ V_{1} \left( \textbf{d}^{l}_{\mathcal{U}_{\textbf{d}}} \right), ..., V_{N_{\textbf{d}}-1} \left( \textbf{d}^{l}_{\mathcal{U}_{\textbf{d}}} \right) \right\} $, where  
\begin{align}\label{MarkovChainLowerBound}
    V_{1} \left( \textbf{d}^{l}_{\mathcal{U}_{\textbf{d}}} \right) \to \cdots \to V_{N_{\textbf{d}}-1} \left( \textbf{d}^{l}_{\mathcal{U}_{\textbf{d}}} \right) \to X \left( \textbf{d}^{l}_{\mathcal{U}_{\textbf{d}}} \right) \to \nonumber\\
    Y_{{{\pi_{{\cal U}_{\textbf{d}}(N_{\textbf{d}})}}}} \left( \textbf{d}^{l}_{\mathcal{U}_{\textbf{d}}} \right) \to \cdots \to Y_{{{\pi_{{\cal U}_{\textbf{d}}(1)}}}} \left( \textbf{d}^{l}_{\mathcal{U}_{\textbf{d}}} \right)
\end{align}
forms a Markov chain, and satisfy
\begin{subequations}
\label{ShirinAppendixLowerBound}
\begin{align}\label{ShirinAppendixLowerBound1}
R - {\varepsilon _n} \le & \frac{1}{n}I\left( {{W_{d^{l}_{\pi_{\mathcal{U}_{{\textbf{d}}}}(1)}}};{U_{\pi_{\mathcal{U}_{{\textbf{d}}}}(1)}}} \right) \nonumber\\
& + I\left( {V}_{{{\cal U}_{\textbf{d}}},1};Y_{\pi_{{{{\cal U}_{\textbf{d}}}}}(1)} \left( \textbf{d}^{l}_{\mathcal{U}_{\textbf{d}}} \right) \right),\\
\label{ShirinAppendixLowerBound2}
R - {\varepsilon _n} \le & \frac{1}{n} I\left( W_{d^{l}_{\pi_{\mathcal{U}_{{\textbf{d}}}}(i)}};U_{\pi_{\mathcal{U}_{{\textbf{d}}}}(1)}, \dots, U_{\pi_{\mathcal{U}_{{\textbf{d}}}}(i)} \right.\nonumber\\
& \qquad \qquad \qquad \left. \left| W_{d^{l}_{\pi_{\mathcal{U}_{{\textbf{d}}}}(1)}}, \dots, W_{d^{l}_{\pi_{\mathcal{U}_{{\textbf{d}}}}(i-1)}} \right. \right) \nonumber\\
&+ I\left( {V}_{i} \left( \textbf{d}^{l}_{\mathcal{U}_{\textbf{d}}} \right); Y_{\pi_{{{{\cal U}_{\textbf{d}}}}}(i)} \left( \textbf{d}^{l}_{\mathcal{U}_{\textbf{d}}} \right) \left| {V}_{i-1} \left( \textbf{d}^{l}_{\mathcal{U}_{\textbf{d}}} \right) \right. \right), \nonumber\\
& \qquad \qquad \qquad \qquad \quad \; \; \forall i \in [2:N_{\textbf{d}}-1],\\
\label{ShirinAppendixLowerBound3}
R - {\varepsilon _n} \le & \frac{1}{n} I\left( W_{d^{l}_{\pi_{\mathcal{U}_{{\textbf{d}}}}(N_{\textbf{d}})}}; U_{\pi_{\mathcal{U}_{{\textbf{d}}}}(1)}, \dots, U_{\pi_{\mathcal{U}_{{\textbf{d}}}}(N_{\textbf{d}})} \right. \nonumber\\
& \qquad \qquad \quad \qquad \left. \left| W_{d^{l}_{\pi_{\mathcal{U}_{{\textbf{d}}}}(1)}}, \dots, W_{d^{l}_{\pi_{\mathcal{U}_{{\textbf{d}}}}(N_{\textbf{d}}-1)}} \right. \right) \nonumber\\
& + I\left( X \left( \textbf{d}^{l}_{\mathcal{U}_{\textbf{d}}} \right); Y_{\pi_{{{{\cal U}_{\textbf{d}}}}}(N_{\textbf{d}})} \left( \textbf{d}^{l}_{\mathcal{U}_{\textbf{d}}} \right) \left| {V}_{N_{\textbf{d}}-1} \left( \textbf{d}^{l}_{\mathcal{U}_{\textbf{d}}} \right) \right. \right),
\end{align}
\end{subequations}
where $d^l_{\pi_{\mathcal{U}_{{\textbf{d}}}}(i)}$ is the $\pi_{\mathcal{U}_{{\textbf{d}}}}(i)$-th element of demand vector $\textbf{d}^{l}_{\mathcal{U}_{\textbf{d}}}$, $i \in \left[ N_{\textbf{d}} \right]$, and ${\varepsilon _n} >0$ tends to zeros as $n \to \infty$. We note that, due to the independence of the files and the fact that the users in ${\cal U}_{\textbf{d}}$ demand distinct files, for any uncoded cache placement phase and any ${\cal U}_{\textbf{d}}$ set, we have 
\begin{align}\label{LowerBoundAppendixMutInfEquality}
    & I\left( W_{d^{l}_{\pi_{\mathcal{U}_{{\textbf{d}}}}(i)}};U_{\pi_{\mathcal{U}_{{\textbf{d}}}}(1)}, \dots, U_{\pi_{\mathcal{U}_{{\textbf{d}}}}(i)} \left| W_{d^{l}_{\pi_{\mathcal{U}_{{\textbf{d}}}}(1)}}, \dots, W_{d^{l}_{\pi_{\mathcal{U}_{{\textbf{d}}}}(i-1)}} \right. \right)\nonumber\\
    &= I\left( W_{d^{l}_{\pi_{\mathcal{U}_{{\textbf{d}}}}(i)}};U_{\pi_{\mathcal{U}_{{\textbf{d}}}}(1)}, \dots, U_{\pi_{\mathcal{U}_{{\textbf{d}}}}(i)} \right), \; \forall i \in [2:N_{\textbf{d}}], l \in [D].
\end{align}
Thus, for an uncoded cache placement phase, \eqref{ShirinAppendixLowerBound} is equivalent to 
\begin{subequations}
\label{ShirinAppendixUncodedCachingLowerBound}
\begin{align}\label{ShirinAppendixUncodedCachingLowerBound1}
R - {\varepsilon _n} \le & \frac{1}{n}I\left( {{W_{d^{l}_{\pi_{\mathcal{U}_{{\textbf{d}}}}(1)}}};{U_{\pi_{\mathcal{U}_{{\textbf{d}}}}(1)}}} \right) \nonumber\\
& + I\left( {V}_{{{\cal U}_{\textbf{d}}},1};Y_{\pi_{{{{\cal U}_{\textbf{d}}}}}(1)} \left( \textbf{d}^{l}_{\mathcal{U}_{\textbf{d}}} \right) \right),\\
\label{ShirinAppendixUncodedCachingLowerBound2}
R - {\varepsilon _n} \le & \frac{1}{n}I\left( W_{d^{l}_{\pi_{\mathcal{U}_{{\textbf{d}}}}(i)}};U_{\pi_{\mathcal{U}_{{\textbf{d}}}}(1)}, \dots, U_{\pi_{\mathcal{U}_{{\textbf{d}}}}(i)} \right) \nonumber\\
&+ I\left( {V}_{i} \left( \textbf{d}^{l}_{\mathcal{U}_{\textbf{d}}} \right); Y_{\pi_{{{{\cal U}_{\textbf{d}}}}}(i)} \left( \textbf{d}^{l}_{\mathcal{U}_{\textbf{d}}} \right) \left| {V}_{i-1} \left( \textbf{d}^{l}_{\mathcal{U}_{\textbf{d}}} \right) \right. \right),\nonumber\\
& \qquad \qquad \qquad \qquad \qquad \forall i \in [2:N_{\textbf{d}}-1],\\
\label{ShirinAppendixUncodedCachingLowerBound3}
R - {\varepsilon _n} \le & \frac{1}{n}I\left( W_{d^{l}_{\pi_{\mathcal{U}_{{\textbf{d}}}}(N_{\textbf{d}})}}; U_{\pi_{\mathcal{U}_{{\textbf{d}}}}(1)}, \dots, U_{\pi_{\mathcal{U}_{{\textbf{d}}}}(N_{\textbf{d}})} \right) \nonumber\\
& + I\left( X \left( \textbf{d}^{l}_{\mathcal{U}_{\textbf{d}}} \right); Y_{\pi_{{{{\cal U}_{\textbf{d}}}}}(N_{\textbf{d}})} \left( \textbf{d}^{l}_{\mathcal{U}_{\textbf{d}}} \right) \left| {V}_{N_{\textbf{d}}-1} \left( \textbf{d}^{l}_{\mathcal{U}_{\textbf{d}}} \right) \right. \right),
\end{align}
\end{subequations}
For the Gaussian channel \eqref{ChannelModel}, for $i=1, ..., N_{\textbf{d}}$, we have \cite{BergmansCapacityDegradeBC}
\begin{align}\label{LowerBoundAppendixBergman}
    & I\left( {V}_{i} \left( \textbf{d}^{l}_{\mathcal{U}_{\textbf{d}}} \right); Y_{\pi_{{{{\cal U}_{\textbf{d}}}}}(i)} \left( \textbf{d}^{l}_{\mathcal{U}_{\textbf{d}}} \right) \left| {V}_{i-1} \left( \textbf{d}^{l}_{\mathcal{U}_{\textbf{d}}} \right) \right. \right) \le \nonumber\\
    & \quad \frac{1}{2}{\log _2}\left( {1 + \frac{{{\beta _{i}} h_{\pi_{{{{\cal U}_{\textbf{d}}}}}(i)}^2 P_{\mathcal{U}_{\textbf{d}}}\left( \textbf{d}^{l}_{\mathcal{U}_{\textbf{d}}} \right)}}{{h_{\pi_{{{{\cal U}_{\textbf{d}}}}}(i)}^2 \sum\nolimits_{j = i + 1}^{N_{\textbf{d}}} {{\beta _j}P_{\mathcal{U}_{\textbf{d}}}\left( \textbf{d}^{l}_{\mathcal{U}_{\textbf{d}}} \right)}  + 1}}} \right),
\end{align}
for some $\beta_i \ge 0$, for $i=1, ..., N_{\textbf{d}}$, such that $\sum\nolimits_{i = 1}^{N_{\textbf{d}}} {{\beta_i}}  = 1$, where we set ${V}_{0} \left( \textbf{d}^{l}_{\mathcal{U}_{\textbf{d}}} \right) \buildrel \Delta \over = 0$, and ${V}_{N_{\textbf{d}}} \left( \textbf{d}^{l}_{\mathcal{U}_{\textbf{d}}} \right) \buildrel \Delta \over = X \left( \textbf{d}^{l}_{\mathcal{U}_{\textbf{d}}} \right)$. From \eqref{ShirinAppendixUncodedCachingLowerBound} and \eqref{LowerBoundAppendixBergman}, for $n$ sufficiently large, the average power $P_{\mathcal{U}_{\textbf{d}}}\left( \textbf{d}^{l}_{\mathcal{U}_{\textbf{d}}} \right)$ to satisfy any demand vector $\textbf{d}^{l}_{\mathcal{U}_{\textbf{d}}} \in \mathcal{D}^l_{\mathcal{U}_{\textbf{d}}}$, for $l \in [D]$, is lower bounded by
\begin{subequations}
\label{LowerBoundPowerMutInfAppendix}
\begin{align}\label{LowerBoundPowerMutInfAppendix1}
P_{\mathcal{U}_{\textbf{d}}}\left( \textbf{d}^{l}_{\mathcal{U}_{\textbf{d}}} \right) \ge & f\left( c^{l}_{\pi_{{{{\cal U}_{\textbf{d}}}}}(1)}, \dots, c^{l}_{\pi_{{{{\cal U}_{\textbf{d}}}}}(N_{\textbf{d}})} \right) \buildrel \Delta \over = \nonumber\\
& \quad \sum\limits_{i = 1}^{N_{\textbf{d}}} \left( \frac{2^{2c^{l}_{\pi_{{{{\cal U}_{\textbf{d}}}}}(i)}} - 1}{h_{\pi_{{{{\cal U}_{\textbf{d}}}}}(i)}^2} \right)\prod\limits_{j = 1}^{i - 1} {{2^{2{c^{l}_{\pi_{{{{\cal U}_{\textbf{d}}}}}(j)}}}}},
\end{align}
where, for $i =1, ..., N_{\textbf{d}}$ and $l=1, ..., D$,
\begin{align}\label{LowerBoundPowerMutInfAppendix2}
c^{l}_{\pi_{{{{\cal U}_{\textbf{d}}}}}(i)} \buildrel \Delta \over = R - \frac{1}{n}I\left( W_{d^{l}_{\pi_{\mathcal{U}_{{\textbf{d}}}}(i)}};U_{\pi_{\mathcal{U}_{{\textbf{d}}}}(1)}, \dots, U_{\pi_{\mathcal{U}_{{\textbf{d}}}}(i)} \right).
\end{align}
\end{subequations}
Note that the lower bound in \eqref{LowerBoundPowerMutInfAppendix1} does not depend on any particular demand in $\mathcal{D}^l_{\mathcal{U}_{\textbf{d}}}$, $l \in [D]$. Thus, from \eqref{SmallerSubsetsPowerMemory2Rewritten} and \eqref{LowerBoundPowerMutInfAppendix}, we have
\begin{align}\label{SmallerSubsetsPowerMemory2FunctionF}
\bar P \ge {{\rm{E}}_{{\mathcal{U}_{\textbf{d}}}}}\left[ \frac{1}{D} \sum\limits_{l=1}^{D} f\left( c^{l}_{\pi_{{{{\cal U}_{\textbf{d}}}}}(1)}, \dots, c^{l}_{\pi_{{{{\cal U}_{\textbf{d}}}}}(N_{\textbf{d}})} \right) \right].
\end{align}

\begin{lemma}\label{DualPowerShirinLemma}
Given a set of users $\mathcal{U}_{\emph{\textbf{d}}}$ of size $N_{\emph{\textbf{d}}}$ with distinct demands, we have
\begin{align}\label{LowerBoundEachDemandLemma}
& \frac{1}{D} \sum\limits_{l=1}^{D} {f\left( c^{l}_{\pi_{{{{\cal U}_{\textbf{d}}}}}(1)}, \dots, c^{l}_{\pi_{{{{\cal U}_{\textbf{d}}}}}(N_{\textbf{d}})} \right)} \ge \nonumber\\
& \sum\limits_{i = 1}^{N_{\emph{\textbf{d}}}}  \left( \frac{{{2^{2R\left( {1 - \min \left\{ {iM/N,1} \right\}} \right)}} - 1}}{h^2_{\pi_{\mathcal{U}_{\emph{\textbf{d}}}}(i)}} \right)\prod\limits_{j = 1}^{i - 1} 2^{2R\left( {1 - \min \left\{ {jM/N,1} \right\}} \right)}.
\end{align}
\end{lemma}
\begin{proof}
It is proved in Appendix \ref{ProofConvexityOfF} that $f\left( \cdot \right)$ is a convex function of $\left( c^{l}_{\pi_{{{{\cal U}_{\textbf{d}}}}}(1)}, \dots, c^{l}_{\pi_{{{{\cal U}_{\textbf{d}}}}}(N_{\textbf{d}})} \right)$. Thus,
\begin{align}\label{ConvexityOfF}
&\frac{1}{D}\sum\limits_{l=1}^{D} {f\left( c^{l}_{\pi_{{{{\cal U}_{\textbf{d}}}}}(1)}, \dots, c^{l}_{\pi_{{{{\cal U}_{\textbf{d}}}}}(N_{\textbf{d}})} \right)} \ge \nonumber\\
& \qquad \qquad f\left( {\frac{1}{D}\sum\limits_{l=1}^{D} c^{l}_{\pi_{{{{\cal U}_{\textbf{d}}}}}(1)} ,...,\frac{1}{D}\sum\limits_{l=1}^{D} c^{l}_{\pi_{{{{\cal U}_{\textbf{d}}}}}(N_{\textbf{d}})} } \right).
\end{align}
From the definition, we have, for $i=1, ..., N_{\textbf{d}}$, 
\begin{align}\label{SimcFromDef}
    & \frac{1}{D}\sum\limits_{l=1}^{D} c^{l}_{\pi_{{{{\cal U}_{\textbf{d}}}}}(i)} = \nonumber\\
    & \qquad R - \frac{1}{nD} \sum\limits_{l=1}^{D} I\left( W_{d^{l}_{\pi_{\mathcal{U}_{{\textbf{d}}}}(i)}};U_{\pi_{\mathcal{U}_{{\textbf{d}}}}(1)}, \dots, U_{\pi_{\mathcal{U}_{{\textbf{d}}}}(i)} \right). 
\end{align}
where, due to the symmetry, each file $W_k$, for $k=1, ..., N$, appears $\binom{N-1}{N_{\textbf{d}}-1}{\left(N_{\textbf{d}}-1\right)!}$ times in the sum on the RHS of \eqref{SimcFromDef} for each $i$ value. Thus, for $i=1, ..., N_{\textbf{d}}$, we have
\newcommand\firstinequality{\mathrel{\overset{\makebox[0pt]{\mbox{\normalfont\tiny\sffamily (a)}}}{\ge}}}
\begin{align}\label{SimcFromDefSim}
    \frac{1}{D}\sum\limits_{l=1}^{D} c^{l}_{\pi_{{{{\cal U}_{\textbf{d}}}}}(i)} & = R - \frac{\binom{N-1}{N_{\textbf{d}}-1}{\left(N_{\textbf{d}}-1\right)!}}{nD} \nonumber\\
    & \qquad \qquad \sum\nolimits_{k=1}^{N} I\left( W_{k};U_{\pi_{\mathcal{U}_{{\textbf{d}}}}(1)}, \dots, U_{\pi_{\mathcal{U}_{{\textbf{d}}}}(i)} \right)\nonumber\\
    & = R - \frac{1}{nN} \sum\nolimits_{k=1}^{N} I\left( W_{k};U_{\pi_{\mathcal{U}_{{\textbf{d}}}}(1)}, \dots, U_{\pi_{\mathcal{U}_{{\textbf{d}}}}(i)} \right)\nonumber\\
    & \firstinequality R - \frac{1}{nN} I\left( W_{1}, ..., W_{N};U_{\pi_{\mathcal{U}_{{\textbf{d}}}}(1)}, \dots, U_{\pi_{\mathcal{U}_{{\textbf{d}}}}(i)} \right) \nonumber\\
    & = R - \frac{R}{N} \min \{ iM,N \}, 
\end{align}
where (a) follows from the independence of the files. From \eqref{ConvexityOfF} and \eqref{SimcFromDefSim}, and the definition of function $f$, we have 
\begin{align}\label{EndOfProofLemma1}
    & \frac{1}{D}\sum\nolimits_{l=1}^{D} {f\left( c^{l}_{\pi_{{{{\cal U}_{\textbf{d}}}}}(1)}, \dots, c^{l}_{\pi_{{{{\cal U}_{\textbf{d}}}}}(N_{\textbf{d}})} \right)} \ge \nonumber\\
    & f\left( R\left( 1-\min \left\{ \frac{iM}{N},1 \right\} \right), \dots, R\left( 1-\min \left\{ \frac{iM}{N},1 \right\} \right) \right),
\end{align}
which concludes the proof of Lemma \ref{DualPowerShirinLemma}. 
\end{proof}

\begin{figure*}[t!]
\centering
\begin{subfigure}{.47\textwidth}
  \centering
  \includegraphics[width=1.0\linewidth]{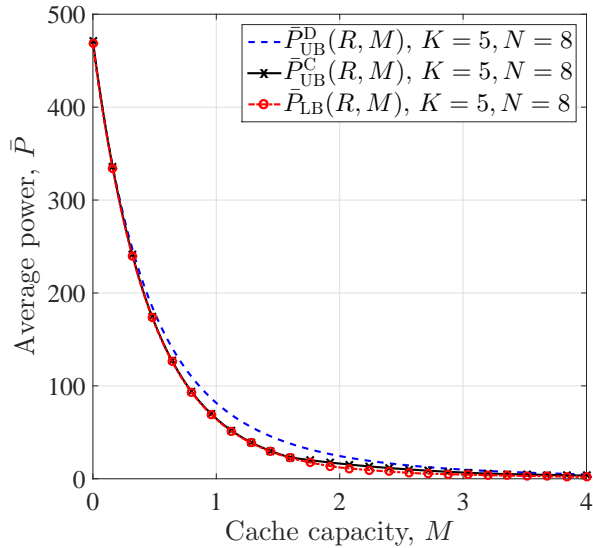}
  \caption{Average power-memory trade-off}
  \label{K5_N8_Average}
\end{subfigure}%
\begin{subfigure}{.47\textwidth}
  \centering
  \includegraphics[width=1.0\linewidth]{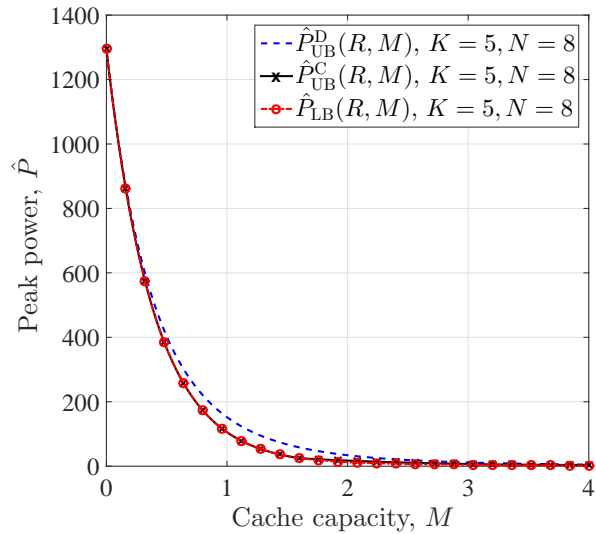}
  \caption{Peak power-memory trade-off}
  \label{K5_N8_Peak}
\end{subfigure}
\caption{Power-memory trade-off for a Gaussian BC with $K=5$ users, and $N=8$ files.}
\label{K5_N8}
\end{figure*}

\begin{figure}[!t]
\centering
\includegraphics[scale=0.4]{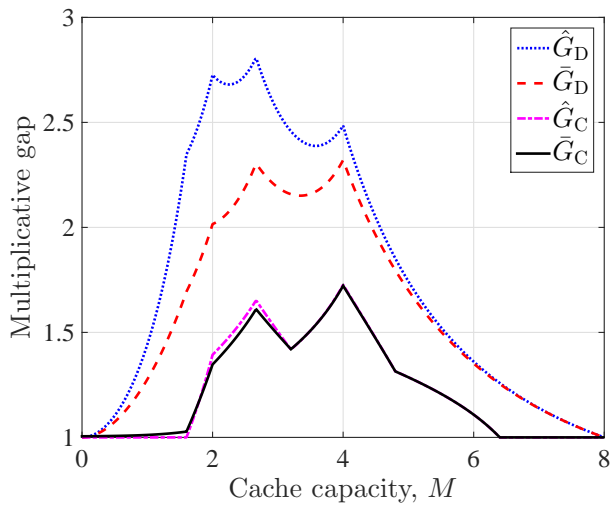}
\caption{Multiplicative gap between the upper and lower bounds for $K=5$ users and $N=8$ files.}
\label{N8_K5_MultGap}
\end{figure}

According to \eqref{SmallerSubsetsPowerMemory2FunctionF} and Lemma \ref{DualPowerShirinLemma}, $\bar P$ is lower bounded by $\bar P_{\rm{LB}}(R,M)$ defined in \eqref{LowerBoundAveragePowerMemoryTheorem}. Thus, $\bar P_{\rm{LB}}(R,M)$ is a lower bound on $\bar P^*_{\mathcal{U}_{\textbf{d}}}(R,M)$ as well as $\bar P^*(R,M)$.

\begin{figure*}[t!]
\centering
\begin{subfigure}{.44\textwidth}
  \centering
  \includegraphics[width=1.0\linewidth]{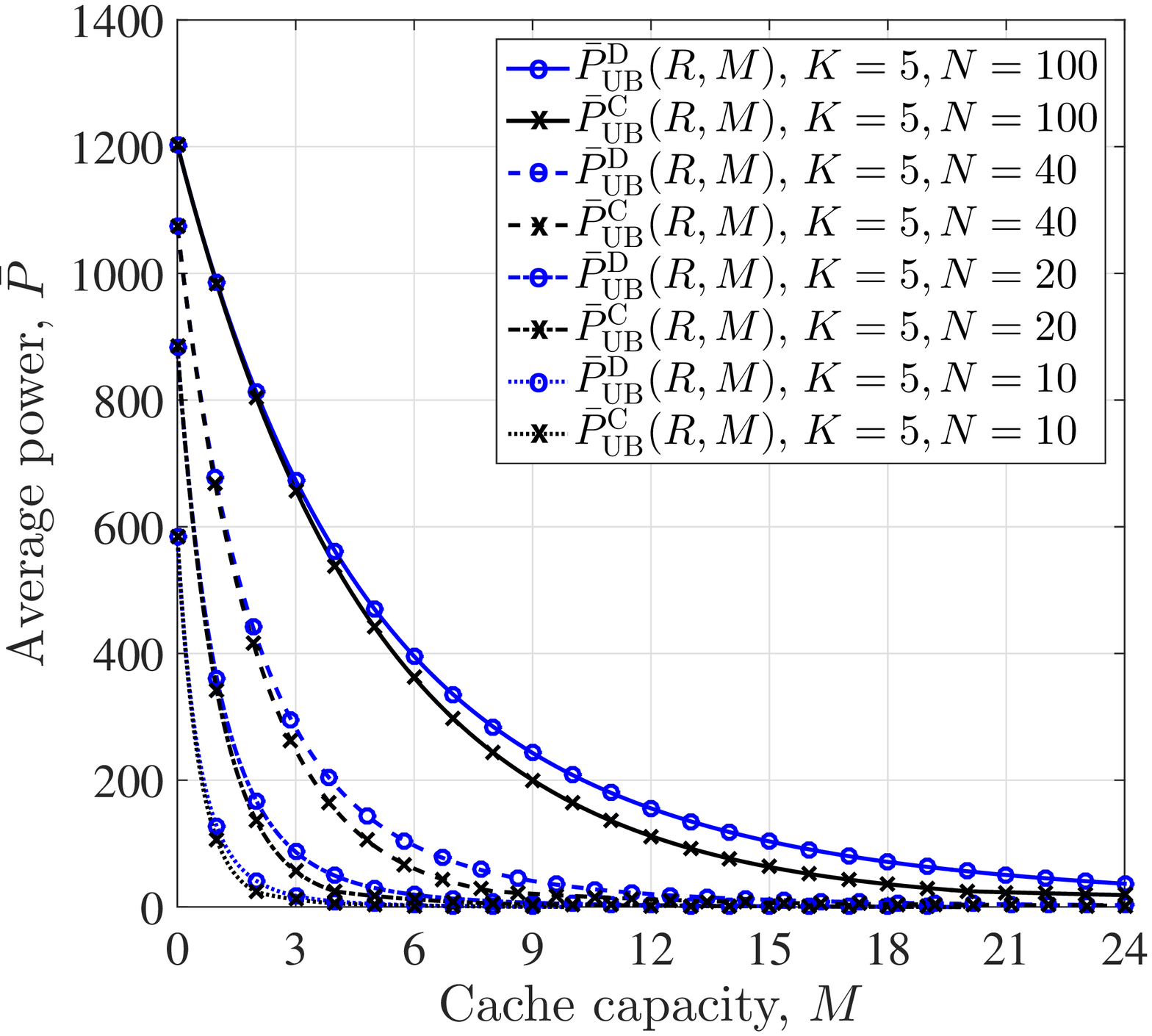}
  \caption{Average power-memory trade-off}
  \label{K5_NVaries_Average}
\end{subfigure}%
\begin{subfigure}{.44\textwidth}
  \centering
  \includegraphics[width=1.0\linewidth]{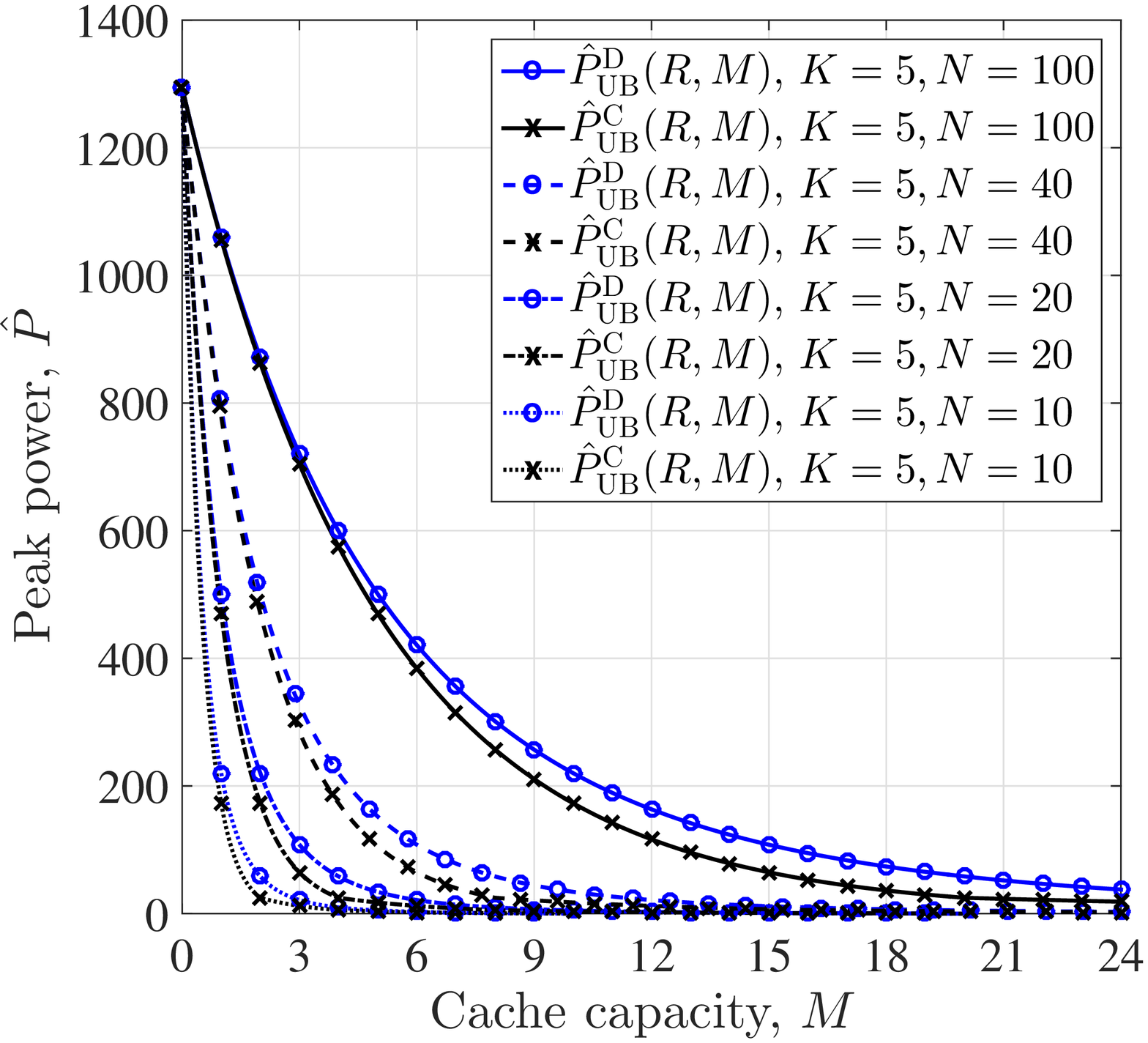}
  \caption{Peak power-memory trade-off}
  \label{K5_NVaries_Peak}
\end{subfigure}
\caption{Power-memory trade-off for a Gaussian BC with $K=5$ users, and various number of files $N \in \{10, 20, 40, 100\}$ in the library.}
\label{K5_NVaries}
\end{figure*}

\begin{figure}[!t]
\centering
\includegraphics[scale=0.39]{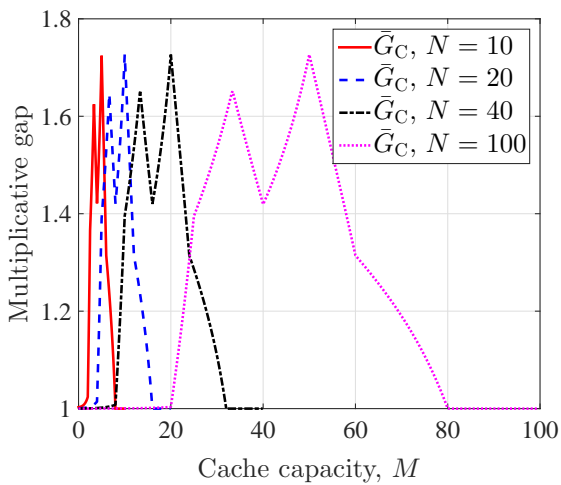}
\caption{The multiplicative gap between the upper and lower bounds on the average power for the centralized caching scenario for $K = 5$ users and $N \in \{ 10,20,40,100\}$ files.}
\label{NVaries_K5_MultGap}
\end{figure}

Next, we prove the lower bound on $\hat P^*(R,M)$ stated in Theorem \ref{LowerboundPowerMemoryTheorem}. For any $\mathcal{U}_{\textbf{d}}$ set, let $\left( R,M,\bar P, \hat P \right)$ tuple be $\mathcal{U}_{\textbf{d}}$-achievable. We have $\hat P \ge {P_{\mathcal{U}_{\textbf{d}}}\left( \textbf{d} \right)}$, $\forall \textbf{d} \in \mathcal{D}_{\mathcal{U}_{\textbf{d}}}$, which is equivalent to
\begin{equation}\label{HatPProof3}
\hat P \ge {P_{\mathcal{U}_{\textbf{d}}}\left( \textbf{d}^l_{\mathcal{U}_{\textbf{d}}} \right)}, \quad \forall \textbf{d}^{l}_{\mathcal{U}_{\textbf{d}}} \in \mathcal{D}^l_{\mathcal{U}_{\textbf{d}}}, \mbox{for $l=1, ..., D$}.  
\end{equation}
Averaging over all $N_{\mathcal{U}_{{\textbf{d}}}} = D N'_{\mathcal{U}_{{\textbf{d}}}}$ possible demands with the same ${\mathcal{U}_{{\textbf{d}}}}$ set, we have
\begin{equation}\label{HatPProof4}
\hat P \ge \frac{1}{D} \sum\nolimits_{l=1}^{D} \left( \frac{1}{N'_{\mathcal{U}_{{\textbf{d}}}}} \sum\nolimits_{\textbf{d}^{l}_{\mathcal{U}_{\textbf{d}}} \in {\mathcal{D}^l_{{\mathcal{U}_{\textbf{d}}}}}} {P_{\mathcal{U}_{\textbf{d}}}\left( \textbf{d}^{l}_{\mathcal{U}_{\textbf{d}}} \right)} \right).  
\end{equation}
According to \eqref{LowerBoundPowerMutInfAppendix} and Lemma \ref{DualPowerShirinLemma}, $\hat P$ is lower bounded as follows:
\begin{align}\label{HatPProof5}
\hat P \ge & \sum\limits_{i = 1}^{N_{{\textbf{d}}}} {\left( \frac{{{2^{2R\left( {1 - \min \left\{ {iM/N,1} \right\}} \right)}} - 1}}{h_{\pi_{\mathcal{U}_{{\textbf{d}}}}(i)}^2} \right)\prod\limits_{j = 1}^{i - 1} 2^{2R\left( {1 - \min \left\{ {jM/N,1} \right\}} \right)} }, \nonumber\\
& \qquad \qquad \qquad \qquad \qquad \qquad \qquad \qquad \qquad  \forall \mathcal{U}_{\textbf{d}}.  
\end{align}
Thus, we have
\begin{align}\label{HatPProof6}
\hat P \ge & \mathop {\max }\limits_{ \mathcal{U}_{\textbf{d}}} \left\{ \sum\nolimits_{i = 1}^{N_{{\textbf{d}}}} \left( \frac{{{2^{2R\left( {1 - \min \left\{ {iM/N,1} \right\}} \right)}} - 1}}{h^2_{\pi_{\mathcal{U}_{{\textbf{d}}}}(i)}} \right) \right. \nonumber\\
& \qquad \qquad \qquad  \left. \prod\nolimits_{j = 1}^{i - 1} 2^{2R\left( {1 - \min \left\{ {jM/N,1} \right\}} \right)}  \right\}.
\end{align}
We note that the term on the RHS of the inequality in \eqref{HatPProof6} is equivalent to $\hat P_{\rm{LB}}(R,M)$ defined in \eqref{LowerBoundWorstCasePowerMemoryTheorem}. Thus, $\hat P_{\rm{LB}}(R,M)$ is a lower bound on ${\hat P}^*_{\mathcal{U}_{\textbf{d}}}\left( {R,M} \right)$ as well as ${{\hat P}^*}\left( {R,M} \right)$. 
\end{proof}

\section{Numerical Results}\label{NumericalRes}
For the numerical results, we assume that the rate of the files in the library is fixed to $R = 1$, and the channel gains are $1/{h_k^2} = 2-0.2(k-1)$, $k \in [K]$.

The average power-memory trade-off $\bar P^*(R,M)$ and peak power-memory trade-off $\hat{P}^*(R,M)$ are shown in Fig. \ref{K5_N8_Average} and Fig. \ref{K5_N8_Peak}, respectively, for $K=5$ users, and $N=8$ files in the library. The gap between centralized and decentralized caching, which measures the power required to compensate for the decentralization of the cache placement phase, is relatively small, particularly for small and large values of the cache capacity. We observe that the minimum average and peak powers drop very quickly even with a small cache capacity available at the users. The lower bound is generally tight for both average and peak power values with respect to the upper bound for the centralized scenario for the whole range of cache capacities\footnote{Note that the figure does not include the cache capacity range $4 \le M \le 8$, in which case the three curves coincide in both figures.}. The peak power values exhibit similar behaviour to the average power, with significantly higher values. This shows that adapting the transmit power to the demand combination can reduce the average energy consumption significantly, while the transmitter spends much higher energy for some demand combinations.

\begin{figure*}[t!]
\centering
\begin{subfigure}{.48\textwidth}
  \centering
  \includegraphics[width=1.0\linewidth]{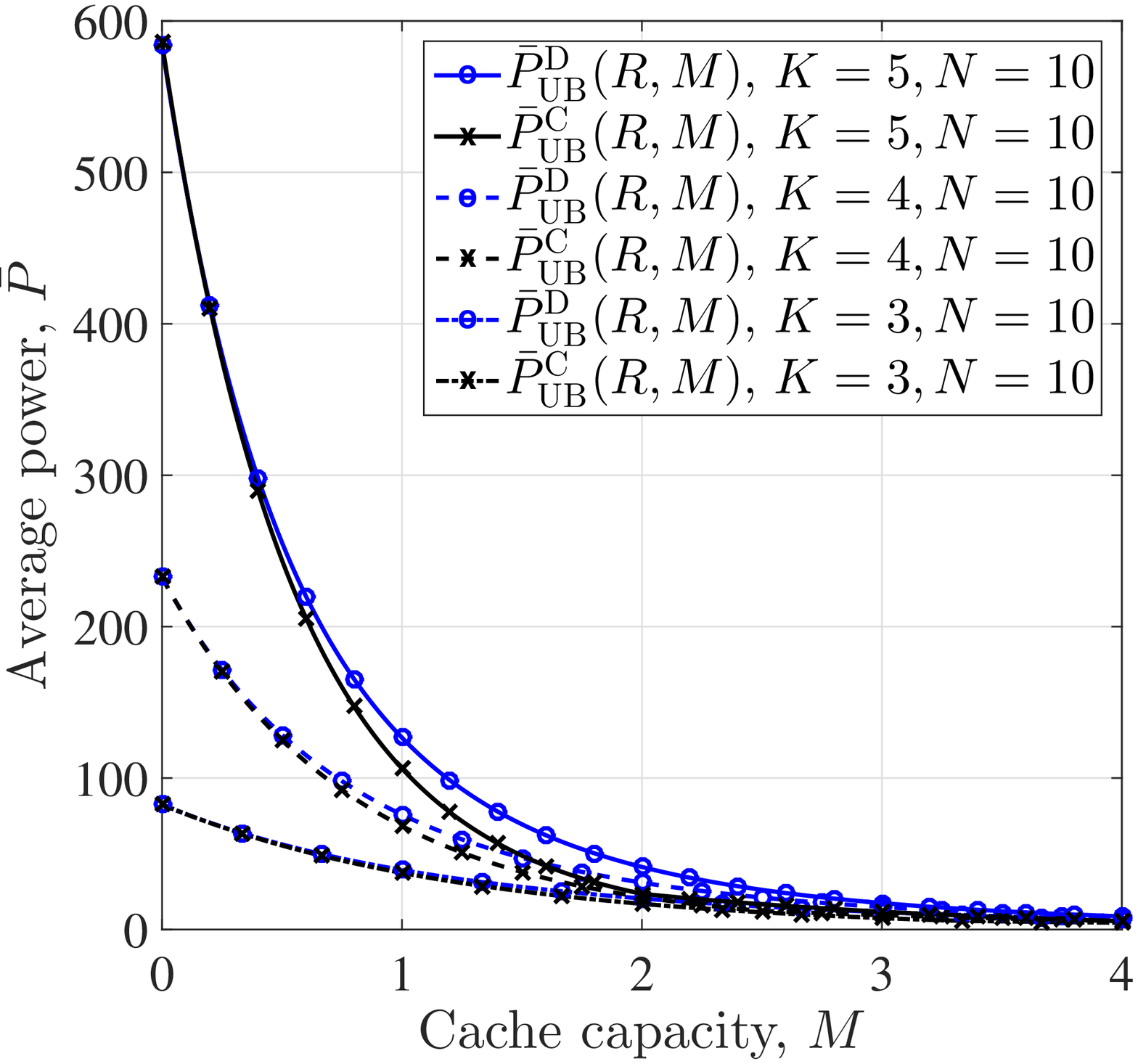}
  \caption{Average power-memory trade-off}
  \label{Kvaries_N10_Average}
\end{subfigure}%
\begin{subfigure}{.48\textwidth}
  \centering
  \includegraphics[width=1.0\linewidth]{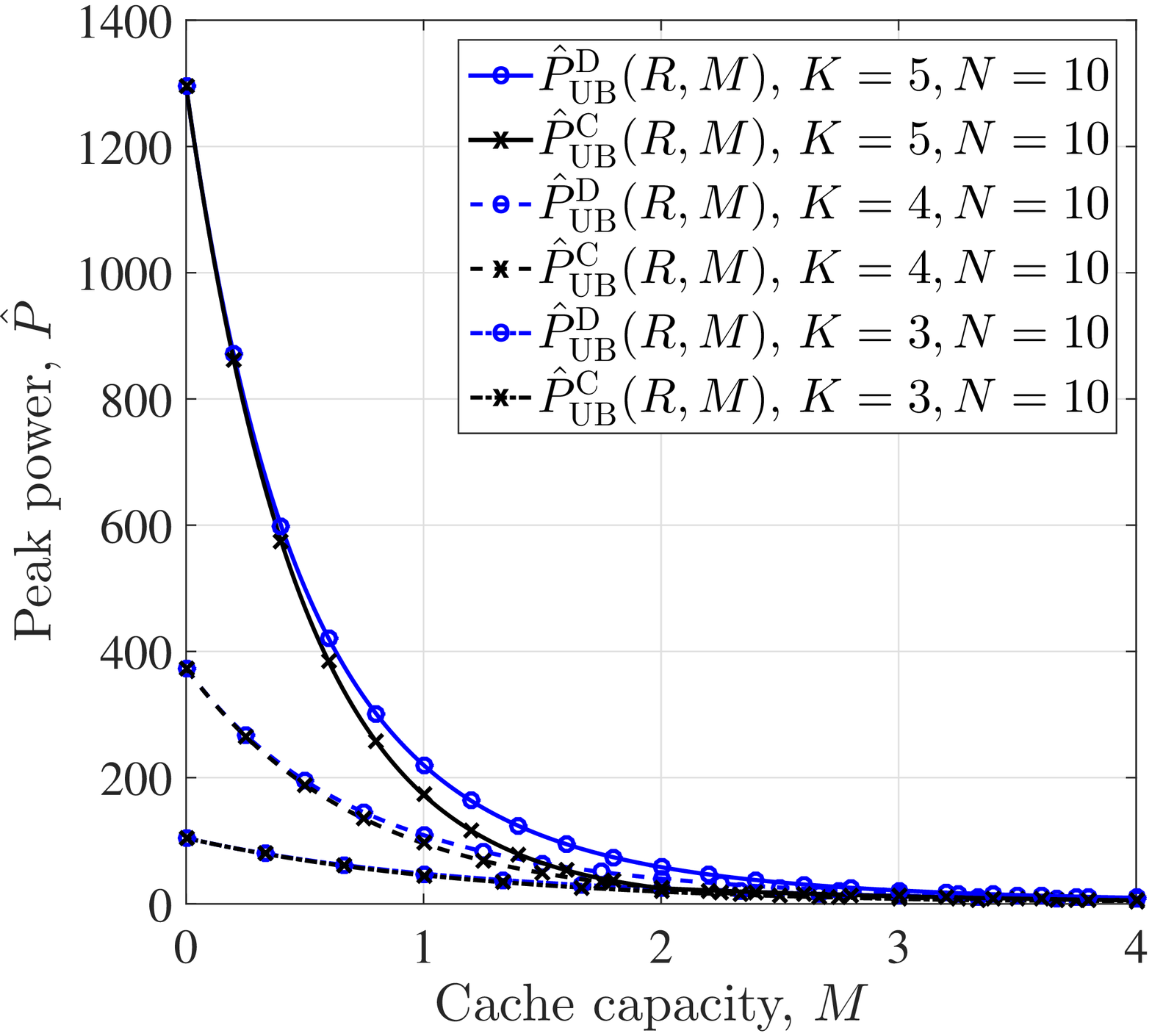}
  \caption{Peak power-memory trade-off}
  \label{Kvaries_N10_Peak}
\end{subfigure}
\caption{Power-memory trade-off for a Gaussian BC with various number of users $K \in \{3,4,5\}$, and $N=10$ in the library.}
\label{Kvaries_N10}
\end{figure*}

\begin{figure}[!t]
\centering
\includegraphics[scale=0.4]{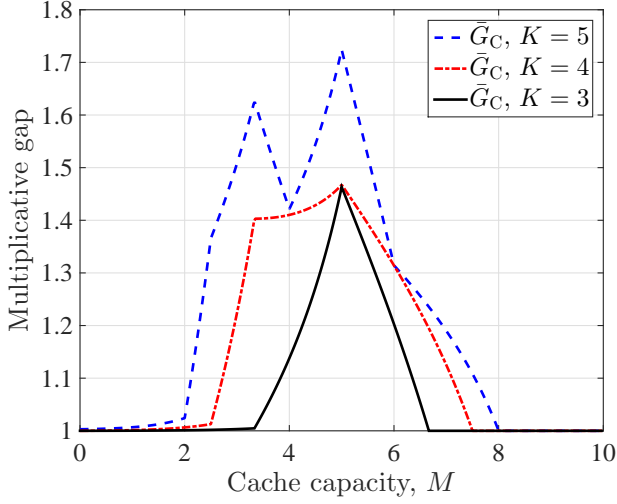}
\caption{The multiplicative gap between the upper and lower bounds on the average power for the centralized caching scenario for $K \in \{ 3,4,5 \}$ users and $N=10$ files.}
\label{N10_KVaries_MultGap}
\end{figure}

We define the following multiplicative gaps between the two bounds on the average and peak power values for the centralized and decentralized caching scenarios:
\begin{subequations}
\begin{align}
\bar{G}_{\rm{C}} \buildrel \Delta \over = \frac{\bar P_{\rm{UB}}^{\rm{C}}(R,M)}{\bar P_{\rm{LB}}(R,M)},\\
\bar{G}_{\rm{D}} \buildrel \Delta \over = \frac{\bar P_{\rm{UB}}^{\rm{D}}(R,M)}{\bar P_{\rm{LB}}(R,M)},\\
\hat{G}_{\rm{C}} \buildrel \Delta \over = \frac{\hat P_{\rm{UB}}^{\rm{C}}(R,M)}{\hat P_{\rm{LB}}(R,M)},\\
\hat{G}_{\rm{D}} \buildrel \Delta \over = \frac{\hat P_{\rm{UB}}^{\rm{D}}(R,M)}{\hat P_{\rm{LB}}(R,M)}.
\end{align}
\end{subequations}
Fig. \ref{N8_K5_MultGap} illustrates the multiplicative gaps for $K=5$ and $N=8$. The multiplicative gap is significantly smaller for centralized caching for both the average  and peak power values. Observe that the multiplicative gap is relatively small for the entire range of cache capacities, and it is higher for the intermediate values of the cache capacity.

Fig. \ref{K5_NVaries} illustrates the effect of the number of files $N$ on the upper bound. In Fig. \ref{K5_NVaries_Average} and Fig. \ref{K5_NVaries_Peak} the average power values and the peak power values are considered, respectively, for $K=5$ users, and different number of files $N \in \{ 10,20,40,100 \}$. The gap between the centralized and decentralized caching scenarios increases $N$. Another observation to be noted is the increase in the average power with $N$, even though the number of files requested by the users remains the same. This stems from two reasons: first, the effect of the local caches diminishes as the library size increases; and second, the users are less likely to request the same files from a larger library (and hence the increase in the average power values for low $M$ values). We also observe that the peak power increases with $N$, but only for non-zero $M$ values. This is because the increase in the peak power is only due to the diminished utility of the cache capacity, whereas the peak power depends only on the worst-case demand combination; and thus, does not depend on the likelihood of common requests. We note that, for a fixed number of users $K$, the gap between the average power and the peak power reduces by increasing $N$, and for $N \gg K$, the gap diminishes, since, with high probability, the users demand distinct files.

Fig. \ref{NVaries_K5_MultGap} illustrates the multiplicative gap $\bar{G}_{\rm{C}}$, i.e., for the upper and lower bounds on the average transmit power, with respect to $M$, for $K=5$ users and $N \in \{ 10,20,40,100 \}$. Observe that the bounds are relatively close, and the multiplicative gap follows a similar pattern for different $N$ values, expanded horizontally with increasing $N$.

The effect of the number of users $K$ on the average and peak power values for both the centralized and decentralized scenarios is shown in Fig. \ref{Kvaries_N10}. Fig. \ref{Kvaries_N10_Average} and Fig. \ref{Kvaries_N10_Peak} demonstrate the average and peak power values as a function of $M$, respectively, for $N = 10$ and $K \in \{ 3,4,5 \}$ users. The gap between centralized and decentralized caching increases with $K$ in both figures. The average and peak power values also increase with $K$, as expected. For a fixed $N$ value, the increase in the peak power is higher than the one in the average power. This is due to the fact that the likelihood of common demands increases with $K$, and so does the gap between the average and peak power values.

Finally, we consider the multiplicative gap $\bar{G}_{\rm{C}}$ for $K \in \{ 3,4,5\}$ and $N=10$ in Fig. \ref{N10_KVaries_MultGap}. It can be seen that the gap increases with $K$.

\section{Conclusions}\label{Conc}
We have considered cache-aided content delivery over a Gaussian BC. Considering same rate contents in the library, we have studied both the \textit{minimum peak transmission power}, which is the minimum transmit power that can satisfy all user demand combinations, and the \textit{minimum average transmit power}, averaged across all demand combinations, assuming uniform demand distributions. We have proposed a centralized caching and coded delivery scheme assuming that the channel conditions in the delivery phase are not known beforehand. Coded contents are transmitted in the delivery phase to their intended receivers using superposition coding and power allocation. We have then extended the achievable scheme to the decentralized caching scenario. We have also provided a lower bound on the required peak and average transmission power values assuming uncoded cache placement. Our results indicate that even a small cache capacity at the receivers can provide a significant reduction in the required transmission power level highlighting the benefits of caching in improving the energy efficiency of wireless networks.

\appendices

\section{Proof of Convexity of Function $f(\cdot)$}\label{ProofConvexityOfF}
We show that, for $1 \le N_{\textbf{d}} \le K$, function $f : \mathbb{R}^{N_{\textbf{d}}} \to \mathbb{R}$ is a convex function of $\left( s_1, ..., s_{N_{\textbf{d}}} \right)$:

\begin{equation}\label{FunctionFrewritten}
    f\left( s_1, \dots, s_{N_{\textbf{d}}} \right) = \sum\nolimits_{i = 1}^{N_{{\textbf{d}}}} {\left( \frac{{{2^{2s_k}} - 1}}{h_{{\pi_{\mathcal{U}_{{\textbf{d}}}}(i)}}^2} \right)\prod\nolimits_{j = 1}^{i - 1} {{2^{2s_j}}} }.
\end{equation}
After mathematical manipulation, one can obtain
\begin{align}\label{FunctionFrewrittenManip}
    & f\left( s_1, \dots, s_{N_{\textbf{d}}} \right) = \nonumber\\
    &  \sum\nolimits_{i = 1}^{N_{{\textbf{d}}}}   \left( \left( \frac{1}{h_{{\pi_{\mathcal{U}_{{\textbf{d}}}}(i)}}^2} - \frac{1}{h_{{\pi_{\mathcal{U}_{{\textbf{d}}}}(i+1)}}^2} \right) 2^{2 \textbf{a}_i^T \textbf{s}} \right) - \frac{1}{h_{{\pi_{\mathcal{U}_{{\textbf{d}}}}(1)}}^2},
\end{align}
where $h_{{\pi_{\mathcal{U}_{{\textbf{d}}}}(N_{\textbf{d}}+1)}}^2 \buildrel \Delta \over = \infty$, $\textbf{s} \buildrel \Delta \over = {\left[ s_1 \; s_2 \; \cdots \; s_{N_{\textbf{d}}} \right]^T}$, and 
\begin{equation}\label{AkDefinition}
    \textbf{a}_i \buildrel \Delta \over = \left[ {\underbrace {1 \; 1 \cdots 1}_{1 \times i}\underbrace {0 \; 0 \cdots 0}_{1 \times \left(  {N_{\textbf{d}} - i} \right)  }} \right]^T, \quad \mbox{for $i=1, ..., N_{\textbf{d}}$}.  
\end{equation}
We note that $h_{{\pi_{\mathcal{U}_{{\textbf{d}}}}(i)}}^2 \le h_{{\pi_{\mathcal{U}_{{\textbf{d}}}}(i+1)}}^2$, $\forall i \in \left[ N_{\textbf{d}} \right]$, and function $2^{2s}$, $s \in \mathbb{R}$, is a convex function of $s$. Thus, all the functions $2^{2 \textbf{a}_i^T \textbf{s}}$, $\forall i \in \left[ N_{\textbf{d}} \right]$, are convex since the affine substitution of the arguments preserves convexity. Hence, function $f$ is convex with respect to $\left( s_1, ..., s_{N_{\textbf{d}}} \right)$ since any linear combination of convex functions with non-negative coefficients is convex.

\section{Proof of Worst-Case Demand Combination for Centralized Caching}\label{ProofWorstCaseDemandCentralized}

We first illustrate that, among all the demand combinations with a fixed $N_{\textbf{d}}$, the worst-case is the one when ${\cal U}_{{\textbf{d}}} = \left[ N_{\textbf{d}} \right]$. Let, for a fixed $N_{\textbf{d}}$, $\mathcal{D}_{N_{\textbf{d}}}^{*}$ denotes the set of demand vectors with ${\cal U}_{{\textbf{d}}} = \left[ N_{\textbf{d}} \right]$. Note that, the required power, given in \eqref{RequiredPowerEachDemandDeliveryCentralized}, is the same for any demand vector $\textbf{d} \in \mathcal{D}_{N_{\textbf{d}}}^{*}$. Assume that there is a demand combination $\tilde{\textbf{d}} \notin \mathcal{D}_{N_{\textbf{d}}}^{*}$, with required transmitted power $\tilde{P}^{\rm{C}} \buildrel \Delta \over = {P}^{\rm{C}}_{\rm{UB}}\left( R,M,\tilde{\textbf{d}} \right)$ determined through \eqref{RequiredPowerEachDemandDeliveryCentralized}, that has the highest transmitted power among all the demand combinations with the same $N_{\textbf{d}}$. For such demand vector $\tilde{\textbf{d}} \notin \mathcal{D}_{N_{\textbf{d}}}^{*}$, there is at least one user $k \in {\cal U}_{\tilde{\textbf{d}}}$ such that $k-1 \notin {\cal U}_{\tilde{\textbf{d}}}$, where we have $N_{\tilde{\textbf{d}},k-1} = N_{\tilde{\textbf{d}},k} + 1$. Now, consider another demand vector $\hat{\textbf{d}}$ with the same entries as demand vector ${\tilde{\textbf{d}}}$ except that $k \notin {\cal U}_{\hat{\textbf{d}}}$ and $k-1 \in {\cal U}_{\hat{\textbf{d}}}$. An example of such demand vector $\hat{\textbf{d}}$ is as follows:
\begin{equation}\label{UserDemandd'''Centralized}
\hat{\textbf{d}}_{i} = 
\begin{cases} 
{\tilde{\textbf{d}}}_{i+1}, &\mbox{if $i=k-1$},\\
{\tilde{\textbf{d}}}_{i-1}, &\mbox{if $i=k$},\\
{\tilde{\textbf{d}}}_{i}, &   \mbox{otherwise}.
\end{cases}
\end{equation}
We denote the required transmitted power to satisfy demand vector $\hat{\textbf{d}}$ by $\hat{P}^{\rm{C}}$, where $\hat{P}^{\rm{C}} \buildrel \Delta \over = {P}^{\rm{C}}_{\rm{UB}}\left( R,M,\hat{\textbf{d}} \right)$. Note that, for demand vector $\hat{\textbf{d}}$, we have $N_{\hat{\textbf{d}},k-1} = N_{\hat{\textbf{d}},k} =  N_{\tilde{\textbf{d}},k}$, and from \eqref{TotalRateUserkCentralized}, 
\begin{equation}\label{EqualRatesCentralized}
R^{\rm{C}}_{\hat{\textbf{d}},i} = R^{\rm{C}}_{{{\tilde{\textbf{d}}}},i}, \quad \forall i \in [K]\backslash \{k-1,k\}.
\end{equation}
Having defined
\begin{subequations}
\label{DefinitionLeadernNonLeadersRates}
\begin{align}\label{DefinitionLeadernNonLeadersRates1}
R_{{\rm{L}},k} \buildrel \Delta \over =& \frac{\binom{K-k}{\left\lfloor t \right\rfloor}}{\binom{K}{\left\lfloor t \right\rfloor}}\left( \left\lfloor t \right\rfloor +1 -t \right)R+\frac{\binom{K-k}{\left\lfloor t \right\rfloor+1}}{\binom{K}{\left\lfloor t \right\rfloor+1}}\left(t- \left\lfloor t \right\rfloor \right)R,
\end{align}

\begin{align}
R_{{\rm{NL}},k} \left( N_{{\textbf{d}},k} \right) \buildrel \Delta \over =& \frac{\binom{K-k}{\left\lfloor t \right\rfloor}-\binom{K-k-N_{{\textbf{d}},k}}{\left\lfloor t \right\rfloor}}{\binom{K}{\left\lfloor t \right\rfloor}}\left( \left\lfloor t \right\rfloor +1 -t \right)R + \nonumber\\
& \frac{\binom{K-k}{\left\lfloor t \right\rfloor+1}-\binom{K-k-N_{{\textbf{d}},k}}{\left\lfloor t \right\rfloor+1}}{\binom{K}{\left\lfloor t \right\rfloor+1}}\left(t- \left\lfloor t \right\rfloor \right)R,
\label{DefinitionLeadernNonLeadersRates2}
\end{align}
\end{subequations}
according to \eqref{RequiredPowerEachDemandDeliveryCentralized}, we have
\begin{align}\label{PowerDifferenceCentralized1}
&\hat{P}^{\rm{C}} - {\tilde{P}}^{\rm{C}} = { {\frac{1}{h_{k-1}^2}\left( {{2^{2{R_{{\rm{L}},k-1}}}} - 1} \right)\prod\nolimits_{j = 1}^{k - 2} {{2^{2{R^{\rm{C}}_{{\tilde{\textbf{d}}},j}}}}} } } \nonumber\\
& + { {\frac{1}{h_{k}^2}\left( {{2^{2{{R_{{\rm{NL}},k} \left( N_{\tilde{{\textbf{d}}},k} \right)}}}} - 1} \right) \left( 2^{ 2{R_{{\rm{L}},k-1}} } \right) \prod\nolimits_{j = 1}^{k - 2} {{2^{2{R^{\rm{C}}_{{\tilde{\textbf{d}}},j}}}}} } } \nonumber\\
& + \sum\nolimits_{i = k+1}^K \frac{1}{h_{i}^2}\left( {{2^{2{R^{\rm{C}}_{{\tilde{\textbf{d}}},i}}}} - 1} \right) \nonumber\\
& \qquad \qquad \prod\limits_{\scriptstyle\;\;\;\,j = 1\hfill\atop
\scriptstyle j \ne k - 1,k\hfill}^{i - 1} \left(2^{2{R^{\rm{C}}_{{\tilde{\textbf{d}}},j}}}\right)\left(2^{2{R_{{\rm{L}},k-1}}}\right)\left(2^{2{R_{{\rm{NL}},k} \left( N_{\tilde{{\textbf{d}}},k} \right)}}\right)\nonumber\\
& - \frac{1}{h_{k-1}^2}\left( {{2^{2{{R_{{\rm{NL}},k-1} \left( N_{\tilde{{\textbf{d}}},k} +1 \right)}}}} - 1} \right)\prod\nolimits_{j = 1}^{k - 2} {{2^{2{R^{\rm{C}}_{{\tilde{\textbf{d}}},j}}}}}   \nonumber\\
& - { {\frac{1}{h_{k}^2}\left( {{2^{2R_{{\rm{L}},k}}} - 1} \right) \left( 2^{2{{R_{{\rm{NL}},k-1} \left( N_{\tilde{{\textbf{d}}},k} +1 \right)}}} \right) \prod\nolimits_{j = 1}^{k - 2} {{2^{2{R^{\rm{C}}_{{\tilde{\textbf{d}}},j}}}}} } }\nonumber\\
& - \sum\nolimits_{i = k+1}^K \frac{1}{h_{i}^2}\left( {{2^{2{R^{\rm{C}}_{{\tilde{\textbf{d}}},i}}}} - 1} \right) \nonumber\\
& \qquad \prod\limits_{\scriptstyle\;\;\;\,j = 1\hfill\atop
\scriptstyle j \ne k - 1,k\hfill}^{i - 1} \left(2^{2{R^{\rm{C}}_{{\tilde{\textbf{d}}},j}}}\right)\left(2^{2{{R_{{\rm{NL}},k-1} \left( N_{\tilde{{\textbf{d}}},k} +1 \right)}}}\right)\left(2^{2R_{{\rm{L}},k}}\right),
\end{align}
where we have used $N_{\tilde{\textbf{d}},k-1} = N_{\tilde{\textbf{d}},k} + 1$, and $N_{\hat{\textbf{d}},k-1} = N_{\hat{\textbf{d}},k} =  N_{\tilde{\textbf{d}},k}$. From \eqref{DefinitionLeadernNonLeadersRates}, we have 
\begin{equation}\label{SimplificationWorstCaseProosCentralized}
{R_{{\rm{NL}},k-1} \left( N_{\tilde{{\textbf{d}}},k} +1 \right)} = {R_{{\rm{NL}},k} \left( N_{\tilde{{\textbf{d}}},k} \right)} + R_{{\rm{L}},k-1}- R_{{\rm{L}},k}. 
\end{equation}
By substituting \eqref{SimplificationWorstCaseProosCentralized} into \eqref{PowerDifferenceCentralized1}, we obtain 
\begin{align}\label{PowerDifferenceCentralized2}
\hat{P}^{\rm{C}} - & {\tilde{P}}^{\rm{C}} = \left( \prod\nolimits_{j = 1}^{k - 2} {2^{2{R^{\rm{C}}_{{\tilde{\textbf{d}}},j}}}} \right) \left( 2^{2{R_{{\rm{L}},k-1}}} \right) \nonumber\\
& \qquad \left( 1 - 2^{2{R_{{\rm{NL}},k} \left( N_{\tilde{\textbf{d}},k} \right) - 2R_{{\rm{L}},k}}} \right) \left( \frac{1}{h_{k-1}^2} - \frac{1}{h_{k}^2} \right).
\end{align}
Note that, 
\begin{align}\label{PowerDifferenceCentralized3}
{R_{{\rm{NL}},k} \left( N_{{\textbf{d}},k} \right) - R_{{\rm{L}},k}} &= - \frac{\binom{K-k-N_{{\textbf{d}},k}}{\left\lfloor t \right\rfloor}}{\binom{K}{\left\lfloor t \right\rfloor}}\left( \left\lfloor t \right\rfloor +1 -t \right)R \nonumber\\
& - \frac{\binom{K-k-N_{{\textbf{d}},k}}{\left\lfloor t \right\rfloor+1}}{\binom{K}{\left\lfloor t \right\rfloor+1}}\left(t- \left\lfloor t \right\rfloor \right)R  \le 0,
\end{align}
and since $h_{k-1}^2 \le h_{k}^2$,
from \eqref{PowerDifferenceCentralized2}, we have $\hat{P}^{\rm{C}} \ge {\tilde{P}}^{\rm{C}}$, which contradicts the assumption that demand vector $\tilde{{\textbf{d}}} \notin \mathcal{D}_{N_{\textbf{d}}}^{*}$ has the highest required transmit power among all the demand combinations with the same $N_{\textbf{d}}$. Therefore, for a fixed $N_{\textbf{d}}$, the demand vectors in $\mathcal{D}_{N_{\textbf{d}}}^{*}$, i.e., the demand vectors with ${\cal U}_{{\textbf{d}}} = \left[ N_{\textbf{d}} \right]$, require the maximum transmit power with the proposed scheme. From \eqref{TotalRateUserkCentralized}, for all the demand vectors in $\mathcal{D}_{N_{\textbf{d}}}^{*}$, user $k \in [K]$ is delivered a message of rate  
\begin{align}\label{SubTotalRateEachUserCentralized}
{R^{*}}^{\rm{C}}_{\textbf{d},k} \left( N_{\textbf{d}} \right) \buildrel \Delta \over = 
\begin{cases} 
\frac{\binom{K-k}{\left\lfloor t \right\rfloor}}{\binom{K}{\left\lfloor t \right\rfloor}}\left( \left\lfloor t \right\rfloor +1 -t \right)R\\
\quad + \frac{\binom{K-k}{\left\lfloor t \right\rfloor+1}}{\binom{K}{\left\lfloor t \right\rfloor+1}}\left(t- \left\lfloor t \right\rfloor \right)R, & \mbox{if $k \in [N_{\textbf{d}}]$},\\
0, & \mbox{otherwise}.
\end{cases}
\end{align}
Thus, according to \eqref{RequiredPowerEachDemandDeliveryCentralized}, for a fixed $N_{\textbf{d}}$, the required power to satisfy any demand combination $\textbf{d} \in \mathcal{D}_{N_{\textbf{d}}}^{*}$ is given by

\begin{align}\label{SubRequiredPowerEachDemandDeliveryCentralized}
P^{\rm{C}}_{\rm{UB}}\left( R,M,\textbf{d} \right) = & \sum\limits_{i = 1}^{N_{\textbf{d}}} {\left( \frac{{{2^{2{{R^{*}}^{\rm{C}}_{\textbf{d},i}\left( N_{\textbf{d}} \right)}}} - 1}}{h_i^2} \right)\prod\limits_{j = 1}^{i - 1} {{2^{2{{R^{*}}^{\rm{C}}_{\textbf{d},j}\left( N_{\textbf{d}} \right)}}}} }.
\end{align}
Observe that, \eqref{SubRequiredPowerEachDemandDeliveryCentralized} is an increasing function of $N_{\textbf{d}}$, and the power is maximized for $N_{\textbf{d}} = \min\{N,K\}$. Thus, the worst-case demand combination for the proposed centralized caching and coded delivery scheme happens when $\mathcal{U}_{\textbf{d}} = \left[ \min\{N,K\} \right]$.

\section{Proof of Worst-Case Demand Combination for Decentralized Caching}\label{ProofWorstCaseDemand}
Similarly to the centralized caching scenario, assume that, for a fixed $N_{\textbf{d}}$, there is a demand combination $\tilde{\textbf{d}} \notin \mathcal{D}_{N_{\textbf{d}}}^{*}$, with required transmitted power $\tilde{P}^{\rm{D}} \buildrel \Delta \over = {P}^{\rm{D}}_{\rm{UB}}\left( R,M,\tilde{\textbf{d}} \right)$ determined through \eqref{RequiredPowerEachDemandDelivery}, that has the highest transmitted power among all the demand combinations with the same $N_{\textbf{d}}$. For such demand vector $\tilde{\textbf{d}} \notin \mathcal{D}_{N_{\textbf{d}}}^{*}$, there is at least one user $k \in {\cal U}_{\tilde{\textbf{d}}}$ such that $k-1 \notin {\cal U}_{\tilde{\textbf{d}}}$, where we have $N_{\tilde{\textbf{d}},k-1} = N_{\tilde{\textbf{d}},k} + 1$. Similarly to \eqref{UserDemandd'''Centralized}, consider another demand vector $\hat{\textbf{d}}$ with the same entries as demand vector ${\tilde{\textbf{d}}}$ except that $k \notin {\cal U}_{\hat{\textbf{d}}}$ and $k-1 \in {\cal U}_{\hat{\textbf{d}}}$. We denote the required transmitted power for demand vector $\hat{\textbf{d}}$ by $\hat{P}^{\rm{D}}$, defined as $\hat{P}^{\rm{D}} \buildrel \Delta \over = {P}^{\rm{D}}_{\rm{UB}}\left( R,M,\hat{\textbf{d}} \right)$. Note that, for demand vector $\hat{\textbf{d}}$, we have $N_{\hat{\textbf{d}},k-1} = N_{\hat{\textbf{d}},k} =  N_{\tilde{\textbf{d}},k}$, and from \eqref{TotalRateEachUser}, 
\begin{equation}\label{EqualRatesDecentralized}
R^{\rm{D}}_{\hat{\textbf{d}},i} = R^{\rm{D}}_{{{\tilde{\textbf{d}}}},i}, \quad \forall i \in [K]\backslash \{k-1,k\}.
\end{equation} 
According to \eqref{RequiredPowerEachDemandDelivery}, we have
\begin{align}\label{PowerDifference1}
&\hat{P}^{\rm{D}} - {\tilde{P}}^{\rm{D}} = { {\frac{1}{h_{k-1}^2}\left( {{2^{2{\left( 1-M/N \right)^{k-1}}R}} - 1} \right)\prod\nolimits_{j = 1}^{k - 2} {{2^{2{R^{\rm{D}}_{{\tilde{\textbf{d}}},j}}}}} } } \nonumber\\
& +  \frac{1}{h_{k}^2}\left( {{2^{2{\left( 1-M/N \right)^{k}}{\left( 1- \left( 1-M/N \right)^{N_{\tilde{\textbf{d}},k}} \right)}R}} - 1} \right)\nonumber\\
& \qquad \qquad \qquad \qquad \left( 2^{2{\left( 1-M/N \right)^{k-1}}R} \right) \prod\nolimits_{j = 1}^{k - 2} 2^{2{R^{\rm{D}}_{{\tilde{\textbf{d}}},j}}}   \nonumber\\
& + \sum\nolimits_{i = k+1}^K \frac{1}{h_{i}^2}\left( {{2^{2{R^{\rm{D}}_{{\tilde{\textbf{d}}},i}}}} - 1} \right) \prod\limits_{\scriptstyle\;\;\;\,j = 1\hfill\atop
\scriptstyle j \ne k - 1,k\hfill}^{i - 1} \left(2^{2{R^{\rm{D}}_{{\tilde{\textbf{d}}},j}}}\right) \nonumber\\
& \qquad \left(2^{2{\left(1-M/N \right)^{k-1}}R}\right) \left(2^{2{\left(1-M/N \right)^{k}}{\left(1-\left(1-M/N \right)^{N_{\tilde{\textbf{d}},k}}\right)}R}\right) \nonumber\\
& - \frac{1}{h_{k-1}^2}\left( {{2^{2{\left( 1-M/N \right)^{k-1}}{\left( 1- \left( 1-M/N \right)^{N_{\tilde{\textbf{d}},k}+1} \right)}R}} - 1} \right)\prod\limits_{j = 1}^{k - 2} {{2^{2{R^{\rm{D}}_{{\tilde{\textbf{d}}},j}}}}}  \nonumber\\
& -  \frac{1}{h_{k}^2}\left( {{2^{2{\left( 1-M/N \right)^{k}}R}} - 1} \right) \nonumber\\
& \qquad \qquad \left( 2^{2{\left( 1-M/N \right)^{k-1}}{\left( 1- \left( 1-M/N \right)^{N_{\tilde{\textbf{d}},k}+1} \right)}R} \right) \prod\limits_{j = 1}^{k - 2} 2^{2{R^{\rm{D}}_{{\tilde{\textbf{d}}},j}}}  \nonumber\\
& - \sum\limits_{i = k+1}^K \frac{1}{h_{i}^2}\left( {{2^{2{R^{\rm{D}}_{{\tilde{\textbf{d}}},i}}}} - 1} \right) \prod\limits_{\scriptstyle\;\;\;\,j = 1\hfill\atop
\scriptstyle j \ne k - 1,k\hfill}^{i - 1} \left(2^{2{R^{\rm{D}}_{{\tilde{\textbf{d}}},j}}}\right) \left(2^{2{\left(1-M/N \right)^{k}}R}\right) \nonumber\\
& \qquad \qquad  \left(2^{2{\left(1-M/N \right)^{k-1}}{\left(1-\left(1-M/N \right)^{N_{\tilde{\textbf{d}},k}+1}\right)}R}\right),
\end{align}
which is equal to
\begin{align}\label{PowerDifference2}
&\hat{P}^{\rm{D}} - {\tilde{P}}^{\rm{D}} = \left( \prod\nolimits_{j = 1}^{k - 2} {2^{2{R^{\rm{D}}_{{\tilde{\textbf{d}}},j}}}} \right) \left( 2^{2{\left( 1-M/N \right)^{k-1}}R} \right. \nonumber\\
& \quad \left. - 2^{2{{\left( 1-M/N \right)^{k-1}}\left( 1- \left( 1-M/N \right)^{N_{\tilde{\textbf{d}},k}+1} \right)}R} \right) \left( \frac{1}{h_{k-1}^2} - \frac{1}{h_{k}^2} \right).
\end{align}
Since $h_{k-1}^2 \le h_{k}^2$, and $0 \le \left( 1-M/N \right)^{N_{\tilde{\textbf{d}},k}+1} \le 1$, from \eqref{PowerDifference2}, we have $\hat{P}^{\rm{D}} \ge {\tilde{P}}^{\rm{D}}$, which contradicts the assumption that demand vector $\tilde{{\textbf{d}}} \notin \mathcal{D}_{N_{\textbf{d}}}^{*}$ has the highest required power among all the demand combinations with the same $N_{\textbf{d}}$. Therefore, for a fixed $N_{\textbf{d}}$, the demand vectors in $\mathcal{D}_{N_{\textbf{d}}}^{*}$, i.e., the demand vectors with ${\cal U}_{{\textbf{d}}} = \left[ N_{\textbf{d}} \right]$, require the maximum transmitted power with the proposed scheme. From \eqref{TotalRateEachUser}, for all the demand vectors in $\mathcal{D}_{N_{\textbf{d}}}^{*}$, we have 
\begin{align}\label{SubTotalRateEachUser}
{R^{*}}^{\rm{D}}_{\textbf{d},k} \left( N_{\textbf{d}} \right) \buildrel \Delta \over = 
\begin{cases} 
{\left( {1 - \frac{M}{N}} \right)^k}R, & \mbox{if $k \in [N_{\textbf{d}}]$},\\
0, & \mbox{otherwise}.
\end{cases}
\end{align}
Thus, according to \eqref{RequiredPowerEachDemandDelivery}, the required power to satisfy any demand combination $\textbf{d} \in \mathcal{D}_{N_{\textbf{d}}}^{*}$ is
\begin{align}\label{SubRequiredPowerEachDemandDelivery}
P^{\rm{D}}_{\rm{UB}}\left( R,M,\textbf{d} \right)  = & \sum\limits_{i = 1}^{N_{\textbf{d}}} {\left( \frac{{{2^{2{\left( 1-M/N \right)^{i}}R}} - 1}}{h_i^2} \right)\prod\limits_{j = 1}^{i - 1} {{2^{2{\left( 1-M/N \right)^{j}}R}}} }.
\end{align}
Observe that, \eqref{SubRequiredPowerEachDemandDelivery} is an increasing function of $N_{\textbf{d}}$, and the required power is maximized for the highest value of $N_{\textbf{d}}$, which is $\min\{N,K\}$. Thus, the worst-case demand combination for the proposed scheme happens when $\mathcal{U}_{\textbf{d}} = \left[ \min\{N,K\} \right]$.

\bibliographystyle{IEEEtran}
\bibliography{Report}

\medskip
 
\noindent\textbf{Mohammad Mohammadi Amiri} (S'16) received the B.Sc. degree (Hons.) in Electrical Engineering from the Iran University of Science and Technology in 2011, and the M.Sc. degree (Hons.) in Electrical Engineering from the University of Tehran in 2014. He is currently pursuing the Ph.D. degree with the Imperial College London. His research interests include information and coding theory, wireless communications, MIMO systems, cooperative networks, cognitive radio, and signal processing.\\

\noindent\textbf{Deniz G\"und\"uz} (S'03-M'08-SM'13) received his M.S. and Ph.D. degrees in electrical engineering from NYU Polytechnic School of Engineering (formerly Polytechnic University) in 2004 and 2007, respectively. After his PhD, he served as a postdoctoral research associate at Princeton University, as a consulting assistant professor at Stanford University, and as a research associate at CTTC. In September 2012, he joined the Electrical and Electronic Engineering Department of Imperial College London, UK, where he is currently a Reader in information theory and communications, and leads the Information Processing and Communications Lab. His research interests lie in the areas of communications and information theory, machine learning, and security and privacy in cyber-physical systems. Dr. G\"und\"uz is an Editor of the IEEE Transactions on Communications, and the IEEE Transactions on Green Communications and Networking. He is the recipient of the IEEE Communications Society - Communication Theory Technical Committee (CTTC) Early Achievement Award in 2017 and a Starting Grant of the European Research Council (ERC) in 2016. He coauthored papers that received a Best Paper Award at the 2016 IEEE WCNC, and Best Student Paper Awards at 2007 IEEE ISIT and 2018 IEEE WCNC.

\end{document}